\documentclass[journal]{IEEEtran}
\usepackage{eurosym}
\usepackage{cite,algorithm,array,url}
\usepackage{graphicx,colortbl}
\usepackage{epstopdf}
\usepackage{amsfonts,times,amsmath,amssymb,amsthm}
\usepackage{multirow}
\usepackage[justification=justified]{caption}
\usepackage{subfigure}

\DeclareMathOperator{\diag}{diag}
\theoremstyle{plain}
\newtheorem{theorem}{Theorem}

\newtheorem{proposition}[theorem]{Proposition}

\theoremstyle{definition}

\theoremstyle{remark}

\input{tcilatex}
\begin{document}

\title{Multiple Object Tracking in Unknown Backgrounds \\
with Labeled Random Finite Sets}
\author{Yuthika Punchihewa, Ba-Tuong Vo, Ba-Ngu Vo and Du Yong Kim\thanks{}}
\maketitle

\begin{abstract}
This paper proposes an on-line multiple object tracking algorithm that can
operate in unknown background. In a majority of multiple object tracking
applications, model parameters for background processes such as clutter and
detection are unknown and vary with time, hence the ability of the algorithm
to adaptively learn the these parameters is essential in practice. In this
work, we detail how the Generalized Labeled Multi-Bernouli (GLMB) filter, a
tractable and provably Bayes optimal multi-object tracker, can be tailored
to learn clutter and detection parameters on-the-fly while tracking.
Provided that these background model parameters do not fluctuate rapidly
compared to the data rate, the proposed algorithm can adapt to the unknown
background yielding better tracking performance.
\end{abstract}

%\markboth{Preprint: IEEE Trans. Signal Processing,~Vol.~65,
%No.~8, pp. 1975--1987, ~April~2017}{Vo et. al. : An Efficient Implementation of the Generalized Labeled
%Multi-Bernoulli Filter}

% The paper headers
%\markboth{Preprint: IEEE Transactions on Signal Processing,~Vol.~XX, No.~X,
%pp. X--X, ~XXX~2014}{Hoang \MakeLowercase{\textit{et al.}}: The Cauchy-Schwarz divergence for Poisson point processes}
% If you want to put a publisher's ID mark on the page you can do it like
% this:
%\IEEEpubid{0000--0000/00\$00.00~\copyright~2012 IEEE}
% Remember, if you use this you must call \IEEEpubidadjcol in the second
% column for its text to clear the IEEEpubid mark.
% use for special paper notices
%\IEEEspecialpapernotice{(Invited Paper)}

% make the title area

% As a general rule, do not put math, special symbols or citations
% in the abstract or keywords.

% Note that keywords are not normally used for peerreview papers.


\section{Introduction}

In a multi-object scenario the number of objects and their individual states
evolve in time, compounded by false detections, misdetections and
measurement origin uncertainty \cite{BSF88, BP99, Mah07, mahler2014advances}%
. For example, in the video dataset KITTI-17 from KITTI datasets \cite{Geiger2012CVPR}, see
Fig. \ref{fig:videosam2}, the number of objects varies with time due to
objects coming in and out of the scene, and the detector (e.g. background
subtraction, foreground modelling \cite{Elgammaletal02}) used to convert
each image into point measurements, invariably misses objects in the scene
as well as generating false measurements or clutter. 

Knowledge of parameters for uncertainty sources such as clutter and
detection profile are of critical importance in Bayesian multi-object
filtering, arguably, more so than the measurement noise model. Most
multi-object tracking techniques are built on the assumption that
multi-object system model parameters are known a priori, which is generally
not the case in practice \cite{BSF88, BP99, Mah07, mahler2014advances}.
Significant mismatches in clutter and detection model parameters inevitably
result in erroneous estimates. For the video tracking example in Fig. \ref%
{fig:videosam2} the clutter rate and detection profile are not known and
have to be guessed before a multi-object tracker can be applied. The
tracking performance of the Bayes optimal multi-object tracking filter \cite%
{VoGLMB13,VVP_GLMB13}, for the guessed clutter rate and 'true' clutter rate (that
varies with time as shown in Fig. \ref{fig:frame_clutter}), demonstrates
significant performance degradation.

\begin{figure*}[h]
\centering
\subfigure{
        \label{fig:subfig1}
    \includegraphics[scale=0.28, trim= 20 0 20 0]{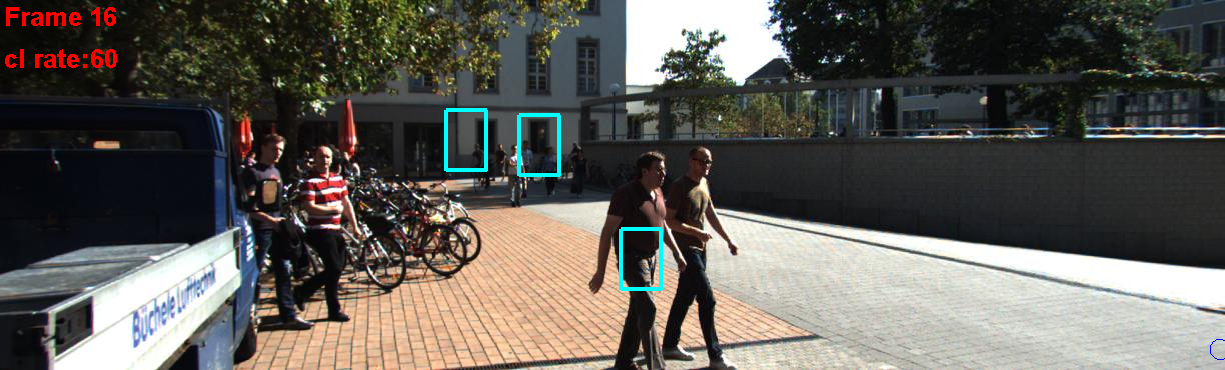}}%
\subfigure{
        \label{fig:subfig1}
    \includegraphics[scale=0.28, trim= 20 0 20 0]{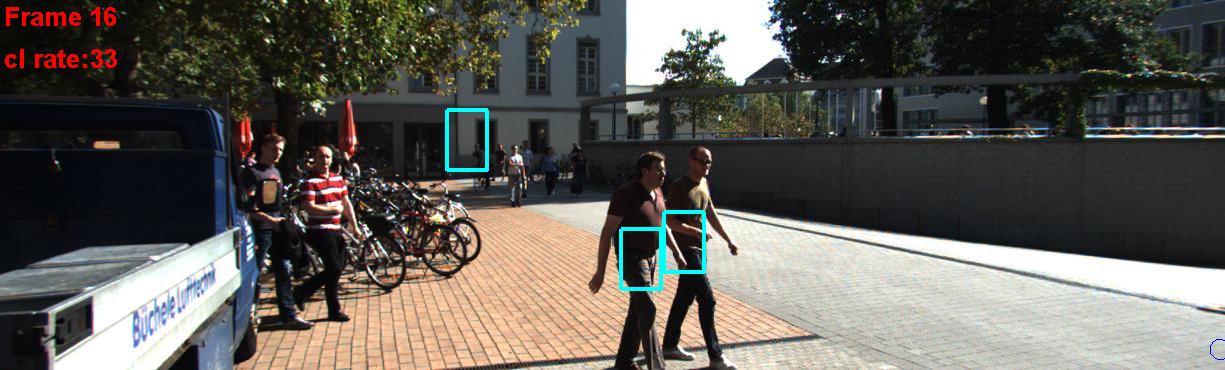}

} 
\subfigure{
        \label{fig:subfig1}
    \includegraphics[scale=0.28, trim= 20 0 20 0]{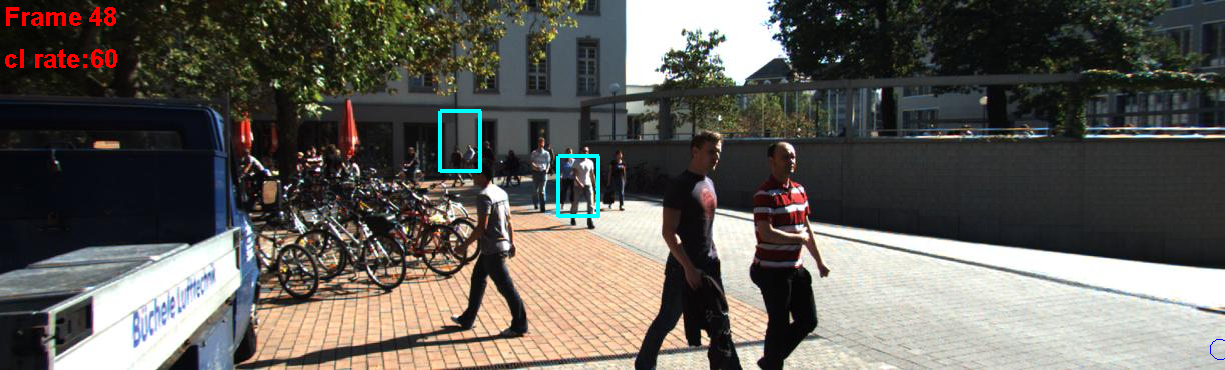}}%
\subfigure{
        \label{fig:subfig1}
    \includegraphics[scale=0.28, trim= 20 0 20 0]{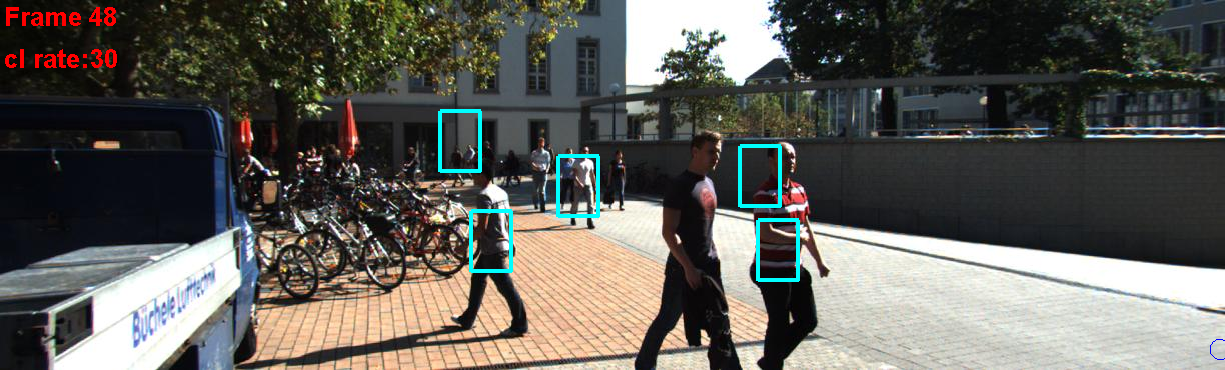}
}
\caption[Optional caption for list of figures]{Frames 16, 48 of the image
sequence from \protect\cite{Geiger2012CVPR} and object detections obtained
using the detector in \protect\cite{DOLLAR2014}. The number of objects
varies with time due to objects coming in and out of the scene. Object
estimates (marked by blue boxes) using the standard GLMB filter for guessed
clutter rate of 60 (left column) and 'true' clutter rate (right column).
Tracking using 'true' clutter rate accurately estimated several objects
that were missed in the frames on the left.}
\label{fig:videosam2}
\end{figure*}

Except for a few applications, the clutter rate and detection profile of the
sensor are not available. Usually these parameters are either estimated from
training data or manually tuned. However, a major problem in many
applications is the time-varying nature of the misdetection and clutter
processes, see Fig. \ref{fig:frame_clutter} for example. Consequently,
there is no guarantee that the model parameters chosen from training data
will be sufficient for the multi-object filter at subsequent frames. Thus,
current multi-object tracking algorithms are far from being a
'plug-and-play' technology, since their application still requires
cumbersome and error-prone user configuration.

\begin{figure}[h]
\centering
\includegraphics[trim=20 0 20 0,clip,width=0.5 \textwidth]{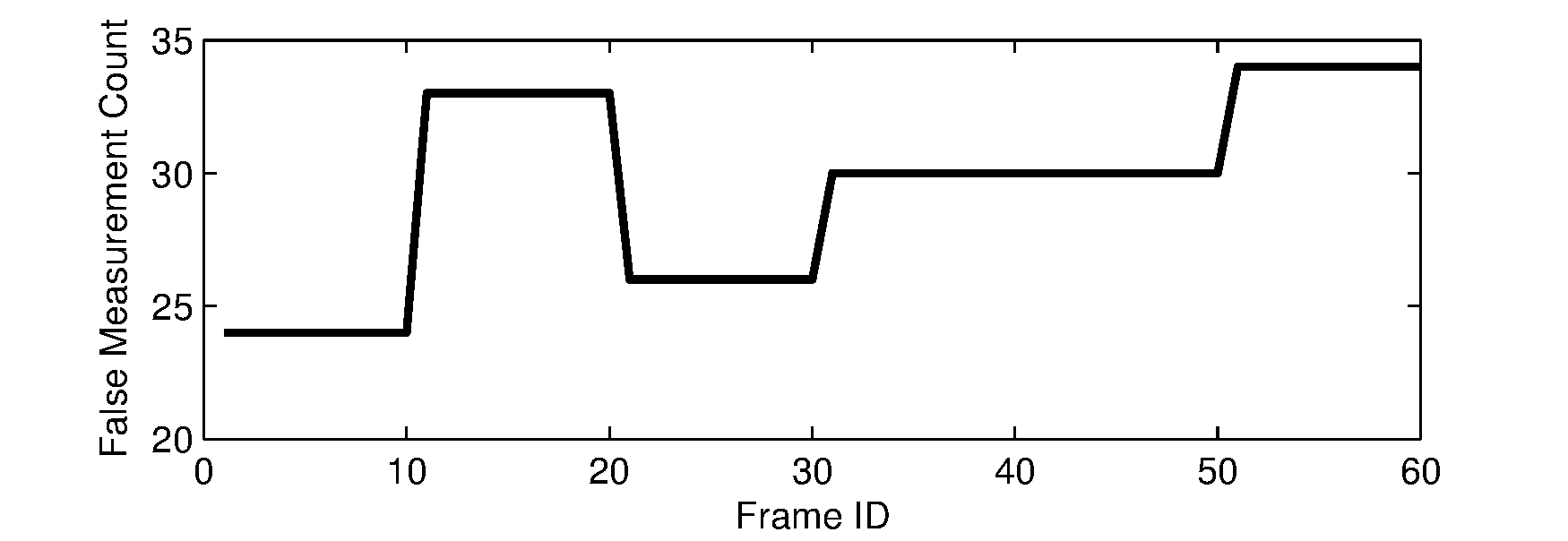}
\caption{'True' clutter rate for the first 60 frames of the dataset\cite{Geiger2012CVPR}. Note
that it is not possible to know the true clutter rate for real video data.
For illustration we assume that the clutter rate varies slowly and
use the average clutter count over a moving 10-frame window as the 'true'
clutter rate.}
\label{fig:frame_clutter}
\end{figure}

This paper proposes an online multi-object tracker that learns the clutter
and detection model parameters while tracking. Such capability is essential
for applications where the clutter rate and detection profile vary with
time. Specifically, we detail a GLMB filter for Jump Markov system (JMS),
which is applicable to tracking multiple manuevering objects as well as
joint tracking and classification of multiple objects. Using the JMS-GLMB
filter, we develop a multi-object tracker that can adaptively learn clutter
rate and detection profile while tracking, provided that the detection
profile and clutter background do not change too rapidly compared to the
measurement-update rate. An efficient implementation of the proposed filter and 
experiments confirm markedly improved performance over existing multi-object
filters for unknown background such as the $\lambda $-CPHD filter \cite%
{MVV10}. Preliminary results have been reported in \cite{PVV16}, which
outlines a GLMB filter for jump-Markov system model.

We remark that robust Bayesian approaches to problems with model mismatch in
the literature such as \cite{Cozman, Noacketal08, Wally, Basu, Berger,
Berliner} are too computationally intensive for an on-line multi-object
tracker. A Sequential Monte Carlo technique for calibration of
time-invariant multi-object model parameters was proposed in \cite{Singhetal}%
. While this approach is quite general it is not directly applicable to
time-varying clutter rate and detection profile, and is also too
computationally intensive for an on-line tracker. Previous work on CPHD/PHD,
multi-Bernoulli and multi-target Bayes filters for unknown clutter rate and
detection profile \cite{MVV10}, \cite{Mahler10a, Mahler10b, VVHM13,
MahlerVo14, Corea16, Rez_TMI15} do not output object tracks. Further, the
CPHD/PHD, multi-Bernoulli filters require more drastic approximations than
the GLMB filter.

The remainder of paper is organized as follows. Section II provides
background material on multi-object tracking and the GLMB filter. Section
III details two versions of the GLMB filter for a general multi-object JMS
model and a non-interacting multi-object JMS model. Section IV presents an
efficient implementation of the non-interacting JMS-GLMB filter for tracking
in unknown clutter rate and detection profile. Numerical studies are
presented in Section V and concluding remarks are given in Section VI.

\section{Background}

\indent This section reviews relevant background on the random finite set
(RFS) formulation of multi-object tracking and the GLMB filter. Throughout
the article, we adopt the following notations. For a given set $S$, $|S|$
denotes its cardinality (number of elements), $1_{S}(\cdot )$ denotes the
indicator function of $S$, and $\mathcal{F}(S)$ denotes the class of finite
subsets of $S$. We denote the inner product $\int f(x)g(x)dx$ by $%
\left\langle f,g\right\rangle $, the list of variables $%
X_{m},X_{m+1},...,X_{n}$ by $X_{m:n}$, the product $\tprod_{x\in X}f(x)$
(with $f^{\emptyset }=1$) by $f^{X}$, and a generalization of the Kroneker
delta that takes arbitrary arguments such as sets, vectors, integers etc.,
by 
\begin{equation*}
\delta _{Y}[X]\triangleq \left\{ 
\begin{array}{l}
1,\text{ if }X=Y \\ 
0,\text{ otherwise}%
\end{array}%
\right. .
\end{equation*}

\subsection{Multi-object State}

\indent At time $k$, an existing object is described by a vector $x_{k}\in 
\mathbb{X}$. To distinguish different object trajectories, each object is
identified by a unique label $\ell _{k}$ that consists of an ordered pair $%
(t,i)$, where $t$ is the time of birth and $i$ is the index of individual
objects born at time $t$ \cite{VoGLMB13}. The trajectory of an object is
given by the sequence of states with the same label.

Formally, the state of an object at time $k$ is a vector $\mathbf{x}%
_{k}=(x_{k},\ell _{k})\in \mathbb{X\times L}_{k}$, where $\mathbb{L}_{k}$
denotes the label space for objects at time $k$ (including those born prior
to $k$). Note that $\mathbb{L}_{k}$ is given by $\mathbb{B}_{k}\cup \mathbb{L%
}_{k-1}$, where $\mathbb{B}_{k}$ denotes the label space for objects born at
time $k$ (and is disjoint from $\mathbb{L}_{k-1}$).

In the RFS approach to multi-object tracking \cite{Mah07, mahler2014advances}%
. the collection of object states, referred to as the \emph{multi-object
state}, is naturally represented as a finite set \cite{VVPS10}. Suppose that
there are $N_{k}$ objects at time $k$, with states $\mathbf{x}_{k,1},...,%
\mathbf{x}_{k,N_{k}}$, then the \emph{multi-object state} is defined by the
finite set 
\begin{equation*}
\mathbf{X}_{k}=\{\mathbf{x}_{k,1},...,\mathbf{x}_{k,N_{k}}\}\in \mathcal{F}(%
\mathbb{X\times L}_{k}),
\end{equation*}%
We denote the set $\{\ell :(x,\ell )\in \mathbf{X}\}$ of labels of $\mathbf{X%
}$ by $\mathcal{L}(\mathbf{X})$. Note that since the label is unique, no two
objects have the same label, i.e. $\delta _{|\mathbf{X}|}[|\mathcal{L}(%
\mathbf{X})|]=1$. Hence $\Delta (\mathbf{X})\triangleq $ $\delta _{|\mathbf{X%
}|}[|\mathcal{L}_{\mathbf{X}}|]$ is called the \emph{distinct label indicator%
}.

A \emph{labeled RFS} is a random variable on $\mathcal{F}(\mathbb{X}\mathcal{%
\times }\mathbb{L})$ such that each realization has distinct labels. The
distinct label property ensures that at any time no two tracks can share any
common points. For the rest of the paper, we follow the convention that
single-object states are represented by lower-case letters (e.g. $x$, $%
\mathbf{x}$), while multi-object states are represented by upper-case
letters (e.g. $X$, $\mathbf{X}$), symbols for labeled states and their
distributions are bold-faced (e.g. $\mathbf{x}$, $\mathbf{X}$, $\mathbf{\pi }
$, etc.), and spaces are represented by blackboard bold (e.g. $\mathbb{X}$, $%
\mathbb{Z}$, $\mathbb{L}$, etc.). For notational compactness, we drop the
time subscript $k$, and use the subscript `$+$' for time $k+1$.

\subsection{Standard multi-object system model}

Given the multi-object state $\mathbf{X}$ at time $k$, each state $(x,\ell
)\in \mathbf{X}$ either survives with probability $P_{S}(x,\ell )$ and
evolves to a new state $(x_{+},\ell _{+})$ at time $k+1$ with probability
density $f_{+}(x_{+}|x,\ell )\delta _{\ell }[\ell _{+}]$ or dies with
probability $1-P_{S}(x,\ell )$. The set $\mathbf{B}_{+}$ of new objects born
at time $k+1$ is distributed according to the labeled multi-Bernoulli (LMB)
density 
\begin{equation}
\Delta (\mathbf{B}_{+})\left[ 1_{\mathbb{B}_{\,+}}\,r_{B,+}\right] ^{%
\mathcal{L(}\mathbf{B}_{+})}\left[ 1-r_{B,+}\right] ^{\mathbb{B}_{+}-%
\mathcal{L(}\mathbf{B}_{+})}p_{B,+}^{\mathbf{B}_{+}},  \label{eq:LMB_birth}
\end{equation}%
where $r_{B,+}(\ell )$ is the probability that a new object with label $\ell 
$ is born, $p_{B,+}(\cdot ,\ell )$\ is the distribution of its kinematic
state, and $\mathbb{B}_{\,+}$ is the label space of new born objects \cite%
{VoGLMB13}. The multi-object state $\mathbf{X}_{+}$ (at time $k+1$) is the
superposition of surviving objects and new born objects. Note that the label
space of all objects at time $k+1$ is the disjoint union $\mathbb{L}_{+}=%
\mathbb{L}\uplus \mathbb{B}_{+}$. It is assumed that, conditional on $%
\mathbf{X}$, objects move, appear and die independently of each other.

For a given multi-object state $\mathbf{X}$, each $(x,\ell )\in \mathbf{X}$
is either detected with probability $P_{D}(x,\ell )$ and generates a
detection $z\in Z$ with likelihood $g(z|x,\ell )$ or missed with probability 
$1-P_{D}(x,\ell )$. The \emph{multi-object observation} is the superposition
of the observations from detected objects and Poisson clutter with
(positive) intensity $\kappa $. Assuming that, conditional on $\mathbf{X}$,
detections are independent of each other and clutter, the multi-object
likelihood function is given by \cite{VoGLMB13}, \cite{VVP_GLMB13} 
\begin{equation}
g(Z|\mathbf{X})\propto \sum_{\theta \in \Theta }1_{\Theta (\mathcal{L(}%
\mathbf{X}))}(\theta )\prod\limits_{(x,\ell )\in \mathbf{X}}\psi
_{Z_{_{\!}}}^{(\theta (\ell ))}(x,\ell )  \label{eq:RFSmeaslikelihood0}
\end{equation}%
where: $\Theta $ is the set of \emph{positive 1-1} maps $\theta :\mathbb{L}%
\rightarrow \{0$:$|Z|\}$, i.e. maps such that \emph{no two distinct
arguments are mapped to the same positive value}, $\Theta (I)$ is the set of 
\emph{positive 1-1} maps with domain $I$; and 
\begin{equation}
\psi _{\!\{z_{1:M}\}\!}^{(j)}(x,\ell )=\left\{ \!\!%
\begin{array}{ll}
\frac{P_{\!D}(x,\ell )g(z_{j}|x,\ell )}{\kappa (z)}, & \!\!\text{if }j=1%
\text{:}M \\ 
1-P_{\!D}(x,\ell ), & \!\!\text{if }j=0%
\end{array}%
\right. \!\!.  \label{eq:PropConj5}
\end{equation}%
The map $\theta $ specifies which objects generated which detections, i.e.
object $\ell $ generates detection $z_{\theta (\ell )}\in Z$, with
undetected objects assigned to $0$. The positive 1-1 property means that $%
\theta $ is 1-1 on $\{\ell :\theta (\ell )>0\}$, the set of labels that are
assigned positive values, and ensures that any detection in $Z$ is assigned
to at most one object.

For the special case with zero-clutter, i.e. $\kappa $ is identically zero,
the multi-object likelihood function still takes the same form, but with $%
P_{\!D}(x,\ell )g(z_{j}|x,\ell )/\kappa (z)$ replaced by $P_{\!D}(x,\ell
)g(z_{j}|x,\ell )$, see \cite{Mah07, mahler2014advances}. To cover both
positive and identically-zero clutter intensities we write%
\begin{equation}
\psi _{\!\{z_{1:M}\}\!}^{(j)}(x,\ell )=\left\{ \!\!%
\begin{array}{ll}
\frac{P_{\!D}(x,\ell )g(z_{j}|x,\ell )}{\kappa (z)+\delta _{0}[\kappa (z)]},
& \!\!\text{if }j=1\text{:}M \\ 
1-P_{\!D}(x,\ell ), & \!\!\text{if }j=0%
\end{array}%
\right. \!\!.
\end{equation}

\subsection{Generalized Labeled Multi-Bernoulli}

A Generalized Labeled Multi-Bernoulli (GLMB) filtering density, at time $k$,
is a multi-object density that can be written in the form 
\begin{equation}
\mathbf{\pi }(\mathbf{X})=\Delta (\mathbf{X})\sum_{\xi \in \Xi ,I\subseteq 
\mathbb{L}}\omega ^{(I,\xi )}\delta _{I}[\mathcal{L(}\mathbf{X})]\left[
p^{(\xi )}\right] ^{\mathbf{X}}.  \label{eq:GLMB}
\end{equation}%
where each $\xi \in \Xi \triangleq \Theta _{0}\times ...\times \Theta _{k}$
represents a history of association maps $\xi =(\theta _{1:k})$, each $%
p^{(\xi )}(\cdot ,\ell )$ is a probability density on $\mathbb{X}$, and each 
$\omega ^{(I,\xi )}$ is non-negative with $\sum_{\xi \in \Xi
}\sum_{I\subseteq \mathbb{L}}\omega ^{(I,\xi )}=1$. The cardinality
distribution of a GLMB is given by%
\begin{equation}
\Pr (\left\vert \mathbf{X}\right\vert \text{=}n)=\sum_{\xi \in \Xi
,I\subseteq \mathbb{L}}\delta _{n}\left[ \left\vert I\right\vert \right]
\omega ^{(I,\xi )},  \label{eq:GLMBCard}
\end{equation}%
while, the existence probability and probability density of track $\ell \in 
\mathbb{L}$ are respectively\allowdisplaybreaks%
\begin{align}
r(\ell )& =\sum_{\xi \in \Xi ,I\subseteq \mathbb{L}}1_{I}(\ell )\omega
^{(I,\xi )}, \\
p(x,\ell )& =\frac{1}{r(\ell )}\sum_{\xi \in \Xi ,I\subseteq \mathbb{L}%
}1_{I}(\ell )\omega ^{(I,\xi )}p^{(\xi )}(x,\ell ).
\end{align}

Given the GLMB density (\ref{eq:GLMB}), an intuitive multi-object estimator
is the \emph{multi-Bernoulli estimator}, which first determines the set of
labels $L\subseteq $ $\mathbb{L}$ with existence probabilities above a
prescribed threshold, and second the mode/mean estimates from the densities $%
p(\cdot ,\ell ),\ell \in L$, for the states of the objects. A popular
estimator is a suboptimal version of the Marginal Multi-object Estimator 
\cite{Mah07}, which first determines the pair $(L,\xi )$ with the highest
weight $\omega ^{(L,\xi )}$ such that $\left\vert L\right\vert $ coincides
with the mode cardinality estimate, and second the mode/mean estimates from $%
p^{(\xi )}(\cdot ,\ell ),\ell \in L$, for the states of the objects.

For the standard multi-object system model the GLMB density is a conjugate
prior, and is also closed under the Chapman-Kolmogorov equation \cite%
{VoGLMB13}. Moreover, the GLMB posterior can be tractably computed to any
desired accuracy in the sense that, given any $\epsilon >0$, an approximate
GLMB within $\epsilon $ from the actual GLMB in $L_{1}$ distance, can be
computed (in polynomial time) \cite{VVP_GLMB13}. The GLMB filtering density
can be propagated forward in time via a prediction step and an update step
as in\cite{VVP_GLMB13} or in one single step as in \cite{VVH_GLMB17}. Since
the number of components grow exponentially in the predicted/filtered
densities during prediction/update stages, truncation of hypotheses with low
weights is essential during implementation. Polynomial complexity schemes
for truncation of insignificant weights were given in \cite{VVP_GLMB13} and 
\cite{VVH_GLMB17}, via Murty's algorithm with a quartic (or at best cubic)
complexity, or via Gibbs sampling with a linear complexity, where the
complexity is given in the number of measurements.

\section{Jump Markov System\ GLMB\ Filtering}

We first derive from the GLMB recursion a multi-object filter for Jump
Markov system (JMS) in subsection \ref{subsec_JMSGLMB}, which is applicable
to tracking multiple manuevering objects as well as joint tracking and
classification of multiple objects. When the modes of the multi-object JMS
do not interact, the JMS-GLMB recursion reduces to a more tractable form,
which is presented in subsection \ref{subsec_MultiClassGLMB}. This special
case is then used to develop a multi-object tracker that can operate in
unknown background in section \ref{sec_UBGLMB}.

\subsection{GLMB filter for Jump Markov Systems}

\label{subsec_JMSGLMB}

A \emph{Jump Markov System} (JMS) consists of a set of parameterised state
space models, whose parameters evolve with time according to a finite state
Markov chain. A JMS can be specified in terms of the standard system
parameters for each mode or class as follows.

Let $\mathbb{M}$ be the (discrete) index set of modes in the system. Suppose
that mode $m$ is in effect at time $k$, then the state transition density
from $\zeta $, at time $k$, to $\zeta _{+}$, at time $k+1$, is denoted by $%
f_{+}^{(m)}(\zeta _{+}|\zeta )$, and the likelihood of $\zeta $ generating
the measurement $z$ is denoted by $g^{(m)}(z|\zeta )$ \cite{LiVSMM00}, \cite%
{LiJilkovMM05}, \cite{Ristic04}. Moreover, the joint transition of the state
and mode assumes the form:%
\begin{equation}
f_{+}(\zeta _{+},m_{+}|\zeta ,m)=f_{+}^{(m_{+})}(\zeta _{+}|\zeta )\vartheta
_{+}(m_{+}|m),  \label{eq:JMS_trans}
\end{equation}%
where $\vartheta _{+}(m_{+}|m)$ denotes the probability of switching from
mode $m$ to $m_{+}$ (and satisfies $\sum_{m_{+}\in \mathbb{M}}\vartheta
_{+}(m_{+}|m)=1$). Note that by defining the \emph{augmented state} as $%
x=(\zeta ,m\mathbf{)}\in \mathbb{X}\times \mathbb{M}$, a JMS model can be
expressed as a standard state space model with transition density (\ref%
{eq:JMS_trans}) and measurement likelihood function $g(z|\zeta
,m)=g^{(m)}(z|\zeta )$.

In a multi-object system, each object is identified by a label $\ell $ that
remains unchanged throughout its life, hence the JMS state equation for such
an object is written as 
\begin{eqnarray}
\!\!\!\!f_{+}(\zeta _{+},m_{+}|\zeta ,m,\ell )\!\!\!
&=&\!\!\!f_{+}^{(m_{+})}(\zeta _{+}|\zeta ,\ell )\vartheta _{+}(m_{+}|m)
\label{eq:JMS-GLMB} \\
\!\!\!\!g(z|\zeta ,m,\ell )\!\!\! &=&\!\!\!g^{(m)}(z|\zeta ,\ell )
\label{eq:JMS-GLMB1}
\end{eqnarray}%
Additionally, to emphasize the dependence on the mode, the survival, birth
and detection parameters are, respectively, denoted as 
\begin{eqnarray*}
p_{B,+}^{(m_{+})}(\zeta _{+},\ell _{+}) &\triangleq &p_{B,+}(\zeta
_{+},m_{+},\ell _{+}), \\
P_{S}^{(m)}(\zeta ,\ell ) &\triangleq &P_{S}(\zeta ,m,\ell ), \\
P_{D}^{(m)}(\zeta ,\ell ) &\triangleq &P_{D}(\zeta ,m,\ell ).
\end{eqnarray*}%
Substituting these parameters and the JMS state equations (\ref{eq:JMS-GLMB}%
)-(\ref{eq:JMS-GLMB1}) into the GLMB recursion in \cite{VVH_GLMB17} yields
the so-called JMS-GLMB recursion.

\begin{proposition}
If the filtering density at time $k$ is the GLMB (\ref{eq:GLMB}), then the
filtering density at time $k+1$ is the GLMB%
\begin{equation}
\mathbf{\pi }_{\!}(\mathbf{X_{+\!}}|Z_{+\!})\!\propto \!\Delta _{\!}(_{\!}%
\mathbf{X_{\!+}}_{\!})\!\!\!\!\sum\limits_{I\!,\xi ,I_{\!+},\theta
_{\!+\!}}\!\!\!\!\omega ^{(I,\xi )}\omega _{Z_{_{\!}+}}^{(_{\!}I_{\!},\xi
,I_{\!+\!},\theta _{\!+\!})}\delta _{_{\!}I_{+\!}}[\mathcal{L}(_{\!}\mathbf{%
X_{\!+}}_{\!})]\!\left[ p_{Z_{_{\!}+_{\!}}}^{(_{\!}\xi ,\theta
_{\!+\!})}{}_{\!}\right] ^{\!\mathbf{X_{\!+}}}\!  \label{eq:JMSGLMB}
\end{equation}%
where $I\in \mathcal{F}(\mathbb{L})$,$\;\xi \in \Xi $,$\;I_{+}\in \mathcal{F}%
(\mathbb{L}_{+})$,$\;\theta _{+}\in \Theta _{+}$, \allowdisplaybreaks%
\begin{align}
\omega _{Z_{_{\!}+}}^{(_{\!}I_{\!},\xi ,I_{\!+\!},\theta _{\!+\!})}=& 1_{{%
\Theta }_{\!+\!}(I_{+})}(\theta _{\!+\!})\left[ 1-\bar{P}_{S}^{(\xi )}\right]
^{\!I\!-I_{\!+}}\!\left[ \bar{P}_{S\!}^{(\xi )}\right] ^{\!I\cap
I_{+\!}}\times  \notag \\
& \left[ 1-r_{B\!,+}\right] ^{\mathbb{B}_{\!+\!}-I_{\!+}\!}\;r_{B\!,+}^{%
\mathbb{B}_{\!{+}}\cap I_{+\!}}\!\left[ \bar{\psi}_{_{\!}Z_{_{\!}+}}^{(_{\!}%
\xi ,\theta _{_{\!}+\!})}\right] ^{I_{+}}  \label{eq:JMSGLMB1} \\
\bar{P}_{S\!}^{(\xi )}(\ell )=& \sum_{m\in \mathbb{M}}\bar{P}_{S\!}^{(\xi
)}(m,\ell ),  \label{eq:JMSGLMB2} \\
\bar{P}_{S\!}^{(\xi )}(m,\ell )=& \left\langle p^{(\xi )\!}(\cdot ,m,\ell
),P_{S}^{(m)}(\cdot ,\ell )\right\rangle ,  \label{eq:JMSGLMB3} \\
\bar{\psi}_{_{\!}Z_{_{\!}+}}^{(\xi ,\theta _{+\!})}(\ell )=& \sum_{m_{+}\in 
\mathbb{M}}\bar{\psi}_{_{\!}Z_{_{\!}+}}^{(\xi ,\theta _{+\!})}(m_{+},\ell ),
\label{eq:JMSGLMB4} \\
\bar{\psi}_{_{\!}Z_{_{\!}+}}^{(\xi ,\theta _{+\!})}(m_{+},\ell )=&
\left\langle \bar{p}_{+}^{(\xi )}(\cdot ,m_{+},\ell _{_{\!}}),\psi
_{_{\!}Z_{_{\!}+}\!}^{(\theta _{_{\!}+}(\ell ))}(\cdot ,m_{+},\ell
_{_{\!}})\right\rangle  \label{eq:JMSGLMB5}
\end{align}%
\begin{gather}
\!\bar{p}_{+}^{(\xi )\!}(\zeta _{+},m_{+},\ell )=1_{\mathbb{B}_{+}}\!(\ell
)p_{B}^{(m_{+})}(\zeta _{+},\ell )\;+\text{ \ \ \ \ \ \ \ \ \ \ \ \ \ \ \ \
\ \ \ \ \ \ }  \notag \\
1_{\mathbb{L}}(\ell )\frac{\sum\limits_{m\in \mathbb{M}}{\tiny \!}%
\!\left\langle \!{\small P}_{S}^{(m)}{\small (\cdot ,\ell )f}_{+}^{(m_{+})}%
{\small (\zeta }_{+}{\small |\cdot ,\ell ),p}^{(\xi )}{\small (\cdot ,m,\ell
)}\!\right\rangle \!{\small \vartheta (m}_{+}{\small |m)}}{{\small \bar{P}}%
_{S}^{(\xi )}{\small (\ell )}}  \label{eq:JMSGLMB6}
\end{gather}%
\begin{align}
p_{Z_{_{\!}+}}^{(\xi _{\!},\theta _{\!+\!})\!}(\zeta _{+},m_{+},\ell )=& 
\frac{\bar{p}_{+}^{(\xi )\!}(\zeta _{+},m_{+},\ell )\psi _{Z_{+}}^{(\theta
_{_{\!}+\!}(\ell ))\!}(\zeta _{+},m_{+},\ell )}{\bar{\psi}_{Z_{+}}^{(\xi
,\theta _{_{\!}+})}(m_{+},\ell _{_{\!}})}  \label{eq:JMSGLMB7} \\
\psi _{\!\{z_{1:|Z|}\}\!}^{(j)}(\zeta ,m,\ell )=& \left\{ \!\!%
\begin{array}{ll}
\!\frac{P_{D}^{(m)}(\zeta ,\ell )g^{(m)}(z_{j\!}|_{_{\!}}\zeta ,\ell )}{%
\kappa (z_{j})+\delta _{0}[\kappa (z_{j})]}, & \!\!\!\text{if }j\in \left\{
1,...,|Z|\right\} \\ 
\!1-P_{D}^{(m)}(\zeta ,\ell ), & \!\!\!\text{if }j=0%
\end{array}%
\right. \!\!  \label{eq:JMSGLMB8}
\end{align}
\end{proposition}

Notice that the above expression is in $\delta $-GLMB form since it can be
written as a sum over $I_{\!+},\xi ,\theta _{+}$ with weights%
\begin{equation*}
\omega _{Z_{+}}^{(I_{+},\xi ,\theta _{+})}\propto \sum\limits_{I}\omega
^{(I,\xi )}\omega _{Z_{+}}^{(I,\xi ,I_{+},\theta _{+})}.
\end{equation*}

This special case of the GLMB recursion is particularly useful for tracking
multiple manuevering objects and joint multi-object tracking and
classification. Indeed the application of the JMS-GLMB recursion to multiple
manuevering object tracking has been reported our preliminary work \cite%
{PVV16}, where separate prediction and update steps was introduced. The same
result was independently reported in \cite{RezaFusion16}.

\subsection{Multi-Class GLMB}

\label{subsec_MultiClassGLMB}

The JMS-GLMB recursion can be applied to the joint multi-object tracking and
classification problem by using the mode as the class label (not to be
confused to object label). What distinguishes this problem from generic
JMS-GLMB filtering is that the modes do not interact with each other in the
following sense:

\begin{enumerate}
\item All possible states of a new object with the same object label share\
a common mode (class label);

\item An object cannot switch between different modes from one time step to
the next.
\end{enumerate}

Let $\mathbb{B}^{(m)}$ denote the set of labels of all elements in $\mathbb{%
X\times M}\times \mathbb{B}$ with mode $m$. Then condition 1 implies that
the label sets $\mathbb{B}^{(m)}$ and $\mathbb{B}^{(m^{\prime })}$ for
different modes $m$ and $m^{\prime }$ are disjoint (otherwise there exist a
label $\ell $ in both $\mathbb{B}^{(m)}$ and $\mathbb{B}^{(m^{\prime })}$,
which means there are states in $\mathbb{X\times M}\times \mathbb{B}$ with
different modes $m$ and $m^{\prime }$ but share a common label $\ell $).
Furthermore, the sets $\mathbb{B}^{(m)}$, $m\in \mathbb{M}$ cover $\mathbb{B}
$, i.e. $\mathbb{B}=\bigcup_{m\in \mathbb{M}}\mathbb{B}^{(m)}$, and thus
form a partition of the space $\mathbb{B}$. A new object is classified as
class $m$ (and has mode $m$) if and only if its label falls into $\mathbb{B}%
^{(m)}$. Thus for an LMB birth\ model, condition 1 means%
\begin{eqnarray}
r_{B,+}(\ell _{+}) &=&\sum_{m_{+}\in \mathbb{M}}r_{B,+}^{(m_{+})}1_{\mathbf{%
\mathbb{B}}_{+}^{(m_{+})}}(\ell _{+}),  \label{eq:class_birthr} \\
p_{B,+}^{(m_{+})}(\zeta _{+},\ell _{+}) &=&p_{B,+}^{(m_{+})}(\zeta _{+})1_{%
\mathbf{\mathbb{B}}_{+}^{(m_{+})}}(\ell _{+}).  \label{eq:class_birthp}
\end{eqnarray}%
Note that $r_{B,+}^{(m_{+})}$ and $p_{B,+}^{(m_{+})}(\zeta _{+})$ are
respectively the existence probability and probability density of the
kinematics $\zeta _{+}$ of a new object given mode $m_{+}$, while $1_{%
\mathbf{\mathbb{B}}_{+}^{(m_{+})}}(\ell _{+})$ is the probability of mode $%
m_{+}$ given label $\ell _{+}$.

Condition 2 means that the mode transition probability 
\begin{equation}
\vartheta (m_{+}|m)=\delta _{m}[m_{+}],  \label{eq:class_trans}
\end{equation}%
which implies that each object belongs to exactly one of the classes in $%
\mathbb{M}$ for its entire life. Consequently, the non-interacting mode
condition means that at time $k$, the label space for all class $m$ objects
is $\mathbb{L}^{(m)}=\biguplus\nolimits_{t=0}^{k}\mathbb{B}_{t}^{(m)}$, and
the set of all possible labels is given by the disjoint union $\mathbb{L}%
=\biguplus\nolimits_{m\in \mathbb{M}}\mathbb{L}^{(m)}$.

For a multi-object JMS system with non-interacting modes, the JMS-GLMB
recursion reduces to a form where the weights and multi-object exponentials
can be separated according to classes. We call this form the multi-class
GLMB.

\begin{proposition}
Let $\mathbf{X}_{\!\!}^{(m)}$ denote the subset of $\mathbf{X}$ with mode $m$%
, and hence $\mathbf{X}=\biguplus\nolimits_{m\in \mathbb{M}}\mathbf{X}^{(m)}$%
. Suppose that the hybrid multi-object density at time $k$ is a GLMB of the
form%
\begin{equation}
\mathbf{\pi }(\mathbf{X})=\sum_{\xi ,I}1_{{\Theta }(I)}(\xi \perp {\Theta }%
)\prod\limits_{m\in \mathbb{M}}\mathbf{\pi }^{(I^{(m)},\xi ^{(m)})}(\mathbf{X%
}_{\!\!}^{(m)})  \label{eq:deltaglmb_class12}
\end{equation}%
where $\xi \in \Xi $, $I\subseteq \mathbb{L}$, $\xi \perp {\Theta }$ denotes
the projection $\xi $ into the space ${\Theta }$, $I^{(m)}\triangleq I\cap 
\mathbb{L}^{(m)}$, $\xi ^{(m)}=\xi |_{\mathbb{L}_{0}^{(m)}\times ...\times 
\mathbb{L}_{k}^{(m)}}$ (i.e. the map $\xi $ restricted to $\mathbb{L}%
_{0}^{(m)}\times ...\times \mathbb{L}_{k}^{(m)}$), and 
\begin{equation}
\mathbf{\pi }^{(I,\xi )}(\mathbf{X})\triangleq \Delta (\mathbf{X})w^{(I,\xi
)}\delta _{_{\!}I}[\mathcal{L}(\mathbf{X}_{\!})]\left[ p^{(\xi )}\right] ^{%
\mathbf{X}}
\end{equation}%
Then the hybrid multi-object filtering density at time $k+1$ is the GLMB 
\begin{equation}
\mathbf{\pi }_{\!Z_{+}\!}(\mathbf{X_{+}})\!\propto \!\!\!\!\sum_{\xi
,I,_{\!}\theta _{\!+\!},I_{+\!}}\!\!\!\!1_{{\Theta }_{\!+\!}(I_{+\!})}(%
\theta _{\!+\!})\!\!\!\prod\limits_{m\in \mathbb{M}}\!\!\mathbf{\pi }%
_{Z_{+}\!}^{\!(m,I_{\!}^{(m)},\xi ^{(m)},I_{\!+\!}^{(m)},\theta
_{\!+\!}^{\!(m)})\!}(\mathbf{X}_{+\!}^{(m)\!})  \label{eq:GLMBactualvirtual}
\end{equation}%
where $I_{\!+\!}\in \!\mathcal{F}(\mathbb{L}_{\!+})$,$\;\theta _{\!+\!}\in
\!\Theta _{+}$, $I_{\!+}^{(m)}=I_{\!+}\!\cap \!\mathbb{L}_{+}^{(m)}$, $%
\theta _{+}^{(m)}=\theta _{+}|_{\mathbb{L}_{+}^{(m)}}$\ \allowdisplaybreaks%
\begin{eqnarray}
&&\!\!\!\!\!\!\!\!\!\!\!\!\!\!\!\!\!\!\!\!\!\!\!\mathbf{\pi }%
_{Z_{+}\!}^{(m,I_{\!},\xi ,I_{\!+\!},\theta _{\!+\!})}(\mathbf{X}_{+\!})= 
\notag \\
&&\!\!\!\!\Delta _{\!}(\mathbf{X}_{+})w_{Z_{_{\!}+}}^{(m,I_{\!},\xi
,I_{\!+\!},\theta _{\!+\!})}w^{(I,\xi )}\delta _{_{\!}I_{+}}[\mathcal{L}(%
\mathbf{X}_{+\!})]\left[ p_{Z_{_{\!}+_{\!}}}^{(_{\!}\xi ,\theta
_{\!+\!})}{}_{\!}\right] ^{\mathbf{X}_{+}}
\end{eqnarray}%
\begin{align}
w_{Z_{_{\!}+}}^{(_{\!}m,I_{\!},\xi ,I_{\!+\!},\theta _{\!+\!})}& =\left[ 
\bar{\psi}_{_{\!}Z_{_{\!}+}}^{(_{\!}\xi ,\theta _{_{\!}+\!})}(m,\cdot )%
\right] ^{I_{+}}\left[ 1-r_{B\!,+}\right] ^{\mathbb{B}_{\!+\!}^{(m)}-I_{\!+}%
\!}r_{B\!,+}^{\mathbb{B}_{\!{+}}^{(m)}\cap I_{+\!}}  \notag \\
& \times \left[ 1-\bar{P}_{S}^{(\xi )}(m,\cdot )\right] ^{\!I\!-I_{\!+}}\!%
\left[ \bar{P}_{S\!}^{(\xi )}(m,\cdot )\right] ^{\!I\cap I_{+\!}}
\label{eq:mclassGLMB2} \\
\bar{P}_{S}^{(\xi )}(m,\ell )& =\left\langle p^{(\xi )\!}(\cdot ,m,\ell
),P_{S}^{(m)}(\cdot ,\ell )\right\rangle ,  \label{eq:mclassGLMB3} \\
\bar{\psi}_{_{\!}Z_{_{\!}+}}^{(\xi ,\theta _{+\!})}(m,\ell )& =\left\langle 
\bar{p}_{+}^{(\xi )}(\cdot ,m,\ell _{_{\!}}),\psi
_{_{\!}Z_{_{\!}+}\!}^{(\theta _{_{\!}+}(\ell ))}(\cdot ,m,\ell
_{_{\!}})\right\rangle ,  \label{eq:mclassGLMB4}
\end{align}%
\begin{align}
\bar{p}_{+}^{(\xi )_{\!}}(\zeta ,m,\ell )=& 1_{\mathbb{L}_{\!}^{(m)}}(\ell )%
\frac{\!\left\langle P_{S}^{(m)}(\cdot ,\ell )f_{+}^{(m)}(\zeta |\cdot ,\ell
)),p^{(\xi )}(\cdot ,m,\ell )\right\rangle }{\bar{P}_{S}^{(\xi )}(m,\ell )} 
\notag \\
+& 1_{\mathbb{B}_{+}^{(m)}}\!(\ell )p_{B}^{(m)}(\zeta ,\ell )
\label{eq:mclassGLMB5} \\
p_{Z_{_{\!}+}}^{(\xi _{\!},\theta _{\!+\!})\!}(\zeta ,m,\ell )& =\frac{\bar{p%
}_{+}^{(\xi )}(\zeta ,m,\ell )\psi _{Z_{+}}^{(\theta _{_{\!}+}(\ell
))}(\zeta ,m,\ell )}{\bar{\psi}_{Z_{+}}^{(\xi ,\theta _{_{\!}+})}(m,\ell )}
\label{eq:mclassGLMB6}
\end{align}

\begin{equation}
\psi _{\!\{z_{1:|Z|}\}\!}^{(j)}(\zeta ,m,\ell )= \Big{\{} 
\begin{array}{ll}
\!\frac{P_{D}^{(m)}(\zeta ,\ell )g^{(m)}(z_{j\!}|_{_{\!}}\zeta ,\ell )}{%
\kappa (z_{j})+\delta _{0}[\kappa (z_{j})]}, \!\text{if }j\in \left\{
1,...,|Z|\right\} &  \\ 
\!1-P_{D}^{(m)}(\zeta ,\ell ), \;\;\;\;\;\; \text{if }j=0 & 
\end{array}%
\end{equation}

%\psi _{\!\{z_{1:|Z|}\}\!}^{(j)}(\zeta ,m,\ell )=& \left\{ \!\!%
%\begin{array}{ll}
%\!\frac{P_{D}^{(m)}(\zeta ,\ell )g^{(m)}(z_{j\!}|_{_{\!}}\zeta ,\ell )}{%
%\kappa (z_{j})+\delta _{0}[\kappa (z_{j})]}, & \!\!\!\text{if }j\in \left\{
%1,...,|Z|\right\} \\ 
%\!1-P_{D}^{(m)}(\zeta ,\ell ), & \!\!\!\text{if }j=0%
%\end{array}
\end{proposition}

\begin{proof}
Note that the $\mathbb{L}_{0}^{(m)}\times ...\times \mathbb{L}%
_{k}^{(m)},m\in \mathbb{M}$ form a partition of $\mathbb{L}_{0}\times
...\times \mathbb{L}_{k}$ , and since each $\xi ^{(m)}$ was defined as a
restrictions of\ $\xi $ over $\mathbb{L}_{0}^{(m)}\times ...\times \mathbb{L}%
_{k}^{(m)}$, $\xi $ is completely characterized by the $\xi ^{(m)},m\in 
\mathbb{M}$. By defining%
\begin{eqnarray}
\omega ^{(I,\xi )} &=&1_{{\Theta }(I)}(\xi \perp {\Theta }%
)\prod\limits_{m\in \mathbb{M}}w^{(I^{(m)},\xi ^{(m)})}  \label{eq:proof_1}
\\
p^{(\xi )}(\zeta ,m,\ell ) &=&\left[ p^{(\xi ^{(m)})}(\zeta ,m,\ell )\right]
^{1_{\mathbb{L}^{(m)}}(\ell )}  \label{eq:proof_3}
\end{eqnarray}%
it can be seen that (\ref{eq:deltaglmb_class12}) is a GLMB of the form (\ref%
{eq:GLMB}) since 
\begin{eqnarray*}
\delta _{_{\!}I}[\mathcal{L}(\mathbf{X}_{\!})] &=&\prod\limits_{m\in \mathbb{%
M}}\delta _{_{\!}I^{(m)}}[\mathcal{L}(\mathbf{X}_{\!}^{(m)})] \\
\left[ p^{(\xi )}\right] ^{\mathbf{X}} &=&\left[ p^{(\xi )}\right]
^{\biguplus\nolimits_{m\in \mathbb{M}}\mathbf{X}^{(m)}}=\prod\limits_{m\in 
\mathbb{M}}\left[ p^{(\xi ^{(m)})}{}_{\!}\right] ^{\mathbf{X}^{(m)}}.
\end{eqnarray*}%
Thus by applying Proposition 1, the hybrid multi-object filtering density at
time $k+1$ is given by (\ref{eq:JMSGLMB}-\ref{eq:JMSGLMB8}). Substituting (%
\ref{eq:proof_1}), (\ref{eq:proof_3}), (\ref{eq:class_birthr}-\ref%
{eq:class_trans}) into (\ref{eq:JMSGLMB}-\ref{eq:JMSGLMB8}), decomposing%
\begin{eqnarray}
\mathbf{X}_{+} &=&\biguplus\nolimits_{m\in \mathbb{M}}\mathbf{X}_{+}^{(m)}
\label{eq:proof1} \\
\omega _{Z_{_{\!}+}}^{(_{\!}I_{\!},\xi ,I_{\!+\!},\theta _{\!+\!})} &=&1_{{%
\Theta }_{\!+\!}(I_{+})}(\theta _{\!+\!})\prod\limits_{m\in \mathbb{M}%
}w_{Z_{_{\!}+}}^{(m,I_{\!}^{(m)},\xi ^{(m)},I_{\!+\!}^{(m)},\theta
_{\!+\!}^{(m)})}  \label{eq:proof2} \\
p_{Z_{_{\!}+_{\!}}}^{(_{\!}\xi ,\theta _{\!+\!})}{} &=&\left(
p_{Z_{_{\!}+_{\!}}}^{(_{\!}\xi ^{(m)},\theta _{\!+\!}^{(m)})}{}_{\!}\right)
^{1_{\mathbb{L}_{+}^{(m)}}(\ell )}  \label{eq:proof3}
\end{eqnarray}%
and rearranging yields (\ref{eq:GLMBactualvirtual}). Note that (\ref%
{eq:class_trans}) ensures that $m_{+}=m$.
\end{proof}

Given a GLMB filtering density of the multi-class form (\ref%
{eq:deltaglmb_class12}), the GLMB filtering density for class $c\in \mathbb{M%
}$, can be obtained by marginalizing the other classes according to the
following proposition.

\begin{proposition}
For the multi-class GLMB (\ref{eq:deltaglmb_class12}), the marginal GLMB for
class $c$ is given by 
\begin{equation*}
\mathbf{\pi }\left( \mathbf{X}^{(c)}\right) =\Delta (\mathbf{X}%
^{(c)})\sum_{\xi ,I}\omega ^{(I,\xi )}\delta _{_{\!}I^{(c)}}[\mathcal{L}(%
\mathbf{X}_{\!}^{(c)})]\left[ p^{(\xi ^{(c)})}\right] ^{\mathbf{X}^{(c)}}
\end{equation*}
\end{proposition}

\begin{proof}
Note that%
\begin{eqnarray*}
&&\!\!\!\!\!\!\!\!\!\!\!\!\!\!\int \mathbf{\pi }^{(I^{(m)},\xi ^{(m)})}(%
\mathbf{X}_{\!\!}^{(m)})\delta \mathbf{X}^{(m)} \\
&=&\!\int \!\Delta (\mathbf{X}^{(m)\!})w^{(I^{(m)\!},\xi ^{(m)})}\delta
_{_{\!}I^{(m)\!}}[\mathcal{L}(\mathbf{X}_{\!}^{(m)})]\!\left[ p^{(\xi )}%
\right] ^{\!\mathbf{X}^{\!(m)}}\!\delta \mathbf{X}^{(m)} \\
&=&\!w^{(I^{(m)},\xi ^{(m)})}.
\end{eqnarray*}%
Since, the $\mathbf{X}^{(m)}$, $m\in \mathbb{M}$ are disjoint, 
\begin{eqnarray*}
\mathbf{\pi }(\mathbf{X}^{(c)})\!\! &=&\!\!\int \mathbf{\pi }\left(
\tbiguplus\limits_{m\in \mathbb{M}}\mathbf{X}^{(m)}\right) \delta \left(
\tbiguplus\limits_{m\in \mathbb{M}-\{c\}}\mathbf{X}^{(m)}\right) \\
\!\! &=&\!\!\int \sum_{\xi ,I}\!1_{{\Theta }(I)}(\xi \!\perp \!{\Theta }%
)\times \\
&&\!\!\prod\limits_{m\in \mathbb{M}}\mathbf{\pi }^{(I^{(m)},\xi ^{(m)})}(%
\mathbf{X}_{\!\!}^{(m)})\delta \left( \tbiguplus\limits_{m\in \mathbb{M}%
-\{c\}}\mathbf{X}^{(m)}\right) \\
\!\! &=&\!\!\sum_{\xi ,I}\!1_{{\Theta }(I)}(\xi \!\perp \!{\Theta })\mathbf{%
\pi }^{(I^{(c)},\xi ^{(c)})}(\mathbf{X}_{\!\!}^{(c)})\times \\
&&\!\!\prod\limits_{m\in \mathbb{M}-\{c\}}\int \mathbf{\pi }^{(I^{(m)},\xi
^{(m)})}(\mathbf{X}_{\!\!}^{(m)})\delta \mathbf{X}^{(m)} \\
\!\! &=&\!\!\sum_{\xi ,I}\!1_{{\Theta }(I)}(\xi \!\perp \!{\Theta })\mathbf{%
\pi }^{(I^{(c)},\xi ^{(c)})}(\mathbf{X}_{\!\!}^{(c)})\times \\
&&\!\!\prod\limits_{m\in \mathbb{M}-\{c\}}\!\!\!\!w^{(I^{(m)},\xi ^{(m)})}.
\\
\!\! &=&\!\!\Delta (\mathbf{X}^{(c)})\sum_{\xi ,I}\omega ^{(I,\xi )}\delta
_{_{\!}I^{(c)}}[\mathcal{L}(\mathbf{X}_{\!}^{(c)})]\left[ p^{(\xi ^{(c)})}%
\right] ^{\mathbf{X}^{(c)}}.
\end{eqnarray*}
\end{proof}

\section{GLMB\ Filtering with Unknown Background} \label{sec_UBGLMB}

Clutter or false detections are generally understood as detections that do
not correspond to any object \cite{BSF88, BP99, Mah07, mahler2014advances}.
Since the number false detections and their values are random, clutter is
usually modelled by RFSs in the literature \cite{MahlerPHD,Mah07,
mahler2014advances}. The simplest and the most commonly used clutter model
is the Poisson RFS \cite{MahlerPHD}, in most cases, with a uniform intensity
over the surveillance region. Alternatively clutter can be treated as
detections originating from \emph{clutter generators}--objects that are not
of interest to the tracker \cite{MVV10}, \cite{Mahler10a, Mahler10b, VVHM13}.

In \cite{MVV10} a CPHD recursion was derived to propagate separate intensity
functions for clutter generators and objects of interest, and their
collective cardinality distribution of the hybrid multi-object state.
Similarly, in\ \cite{VVHM13} analogous multi-Bernoulli recursions were
derived to propagate the disjoint union of objects of interest and clutter
generators. In this work we show that the multi-class GLMB filter is an
effective multi-object object tracker that can operate under unknown
background by learning the clutter and detection model on-the-fly.

This section details an on-line multi-object tracker that operates in
unknown clutter rate and detection profile. In particular we propose a GLMB
clutter model in subsection \ref{subsec_hybridmodel} by treating clutter as
a special class of objects with completely uncertain dynamics, and describe
a dedicated GLMB recursion for propagating the joint filtering density of
clutter generators and objects of interest. Implementation details are given
in subsection \ref{subsec_implementation}. Extension of the proposed
algorithm to accommodate unknown detection profile is described in
subsection \ref{subsec_UPD}.
\vspace{-3mm}
\subsection{GLMB Joint Object-Clutter Model} \label{subsec_hybridmodel}

We propose to model the finite set of \emph{clutter generators} and \emph{%
objects of interest} as two non-interacting classes of objects, and
propagate this so-called \emph{hybrid multi-object} filtering density
forward in time via the multi-class GLMB recursion. The GLMB filtering
density of the hybrid multi-object state captures all relevant statistical
information on the objects of interest as well as the clutter generators.
What distinguishes the objects of interest from clutter generators is that
the former have relatively predictable dynamics whereas the latter have
completely random dynamics.

In the hybrid multi-object model, the Poisson clutter intensity $\kappa $ is
identically 0 and each detection is generated from either a clutter
generator or an object of interest, which constitute, respectively, the two
modes (or classes) 0 and 1 of the mode space $\mathbb{M}=\{0,1\}$. Since the
classes are non-interacting, there are no switchings between objects of
interest and clutter generators. Moreover, the label space for new born
clutter generators $\mathbf{\mathbb{B}}^{(0)}$ and the label space for new
born objects of interest $\mathbf{\mathbb{B}}^{(1)}$ are disjoint and the
LMB birth parameters are given by 
\begin{eqnarray*}
r_{B,+}(\ell _{+}) &=&r_{B,+}^{(0)}1_{\mathbf{\mathbb{B}}_{+}^{(0)}}(\ell
_{+})+r_{B,+}^{(1)}1_{\mathbf{\mathbb{B}}_{+}^{(1)}}(\ell _{+}), \\
p_{B,+}^{(m_{+})}(\zeta _{+},\ell _{+}) &=&p_{B,+}^{(m_{+})}(\zeta _{+})1_{%
\mathbf{\mathbb{B}}_{+}^{(m_{+})}}(\ell _{+})
\end{eqnarray*}%
Since clutter are distinguishable from targets by their completely random
dynamics, each clutter generator has a transition density independent of the
previous state and a uniform measurememt likelihood in the observation
region with volume $V$%
\begin{eqnarray*}
f_{+}^{(0)}(\zeta _{+}|\zeta ,\ell ) &=&s(\zeta _{+}) \\
g^{(0)}(z|\zeta ,\ell ) &=&u(z)V^{-1}
\end{eqnarray*}%
Note that the labels of clutter generators can effectively be ignored since
it is implicit that their labels are distinct but are otherwise
uninformative. Further, for Gaussian implementations it is assumed that the
survival and detection probabilities for clutter generators are state
independent%
\begin{eqnarray*}
P_{S}^{(0)}(\zeta ,\ell ) &=&P_{S}^{(0)} \\
P_{D}^{(0)}(\zeta ,\ell ) &=&P_{D}^{(0)}
\end{eqnarray*}

Applying the multi-class GLMB recursion to this model, it can be easily seen
that all clutter generators are functionally identical (from birth through
prediction and update)%
\begin{equation*}
p_{B}^{(0)}(\zeta ,\ell )=\bar{p}_{+}^{(\xi ^{(0)})}(\zeta ,0,\ell
)=p_{Z_{+}}^{(\xi ^{(0)},\theta _{+}^{(0)})}(\zeta ,0,\ell )=s(\zeta )
\end{equation*}%
and that the weight update for clutter generators reduces to%
\begin{align}
& \!\!\!\!\!\!w_{Z_{_{\!}+}}^{(0,_{\!}I_{\!}^{(0)},\xi
^{(0)},I_{\!+\!}^{(0)},\theta _{\!+\!}^{(0)})}  \notag \\
=& \left[ 1-P_{S}^{(0)}\right] ^{|I^{(0)}-I_{+}^{(0)}|}\left[ P_{S}^{(0)}%
\right] ^{|I^{(0)}\cap I_{+}^{(0)}|}\times  \notag \\
& \left[ 1-r_{B,+}^{(0)}\right] ^{|\mathbb{B}_{+}^{(0)}-I_{+}^{(0)}|}\left[
r_{B,+}^{(0)}\right] ^{|\mathbb{B}_{{+}}^{(0)}\cap I_{+}^{(0)}|}\times 
\notag \\
& \left[ 1-P_{D,+}^{(0)}\right] ^{|\{\ell \in I_{+}^{(0)}:\theta
_{+}^{(0)}(\ell )=0\}|}\left[ P_{D,+}^{(0)}V^{-1}\right] ^{|\{\ell \in
I_{+}^{(0)}:\theta _{+}^{(0)}(\ell )>0\}|}  \label{eq:CGprop}
\end{align}%
Thus propagation of clutter generators within each GLMB\ component reduces
to propagation of their weights 
\begin{equation*}
w_{Z_{+}}^{(0,I_{+}^{(0)},\xi ^{(0)},\theta
_{+}^{(0)})}=\tsum\nolimits_{I^{(0)}}w^{(I^{(0)},\xi
^{(0)})}w_{Z_{+}}^{(0,I^{(0)},\xi ^{(0)},I_{+}^{(0)},\theta _{+}^{(0)})}.
\end{equation*}
\vspace{-3mm}
\subsection{Implementation} \label{subsec_implementation}

The key challenge in the implemention of the multi-class GLMB filter is the
propagation of the GLMB components, which involves, for each parent GLMB
component $(I,\xi )$, searching the space $\mathcal{F}(\mathbb{L}%
_{\!+})\times \Theta _{+}$ to find a set of $(I_{\!+},\theta _{\!+\!})$ such
that the children components $(I_{\!},\xi ,I_{\!+\!},\theta _{\!+\!})$ have
significant weights $\omega _{Z_{_{\!}+}}^{(_{\!}I_{\!},\xi
,I_{\!+\!},\theta _{\!+\!})}$. In \cite{VVH_GLMB17}, the set of $%
(I_{\!+},\theta _{\!+\!})$ is generated from a Gibbs sampler with stationary
distribution is constructed so that only valid children components have
positive probabilities, and those with high weights are more likely to be
sampled than those with low weights. A direct application of this approach
to generate new children would, however, be expensive, for the following
reasons.

Let $P=|I|$, $P^{(0)}=|I^{(0)}|$, $P^{(1)}=|I^{(1)}|$ and $M=|Z_{+}|$.
According to \cite{VVH_GLMB17} the complexity of the joint prediction and
update via Gibbs sampling with $T$ iterations is $\mathcal{O}(TP^{2}M)$.
Since the present formulation treat clutter as objects, the total number of
hypothesized objects $P\geq P^{(0)}\geq M$, and hence the complexity is at
least $\mathcal{O}(TM^{3})$, which is cubic in the number of measurements
and results in a relatively inefficient implementation. This occurs because
the majority of the computational effort is spent on clutter generators even
though they are not of interest. This problem is exacerbated as the clutter
rate increases.

In the following we propose a more efficient implementation by focusing on
the filtering density of the objects of interest instead of the hybrid
multi-object filtering density. Observe that given any $(I_{+}^{(1)},\theta
_{\!+\!}^{(1)})\in \mathcal{F}(\mathbb{L}_{\!+}^{(1)})\times \Theta
_{+}^{(1)}$, and $(I_{+}^{(0)},\theta _{\!+\!}^{(0)})\in \mathcal{F}(\mathbb{%
L}_{\!+}^{(0)})\times \Theta _{+}^{(0)}$, where $\Theta _{+}^{(m)}$\ denotes
the space of positive 1-1 maps from $\mathbb{L}_{+}^{(m)}$ to $\{0,1,...,M\}$%
, we can uniquely define%
\begin{equation}
(I_{\!+},\theta _{\!+\!})\triangleq (I_{+}^{(1)}\uplus I_{+}^{(0)},1_{%
\mathbb{L}_{\!+}^{(1)}}\theta _{\!+\!}^{(1)}+1_{\mathbb{L}%
_{\!+}^{(0)}}\theta _{\!+\!}^{(0)}).  \label{eq:composition}
\end{equation}%
Further, the weight of the resulting component $(I_{\!},\xi
,I_{\!+\!},\theta _{\!+\!})$ is 
\begin{equation}
\omega _{Z_{_{\!}+}}^{(_{\!}I_{\!},\xi ,I_{\!+\!},\theta _{\!+\!})}=1_{{%
\Theta }(I_{\!+\!})}(\theta _{\!+\!})w_{Z_{+}\!}^{\!(0,I_{\!}^{(0)},\xi
^{(0)},I_{\!+\!}^{(0)},\theta
_{\!+\!}^{\!(0)})\!}w_{Z_{+}\!}^{\!(1,I_{\!}^{(1)},\xi
^{(1)},I_{\!+\!}^{(1)},\theta _{\!+\!}^{\!(1)})\!}
\label{eq:componentweight}
\end{equation}%
see Proposition 2 (\ref{eq:proof2}). Note that if $\theta _{\!+\!}$ is not a
valid association map then $1_{{\Theta }(I_{\!+\!})}(\theta _{\!+\!})=0$,
and hence the weight is zero.

For each parent GLMB component $(I,\xi )$, rather than searching for $%
(I_{\!+},\theta _{\!+\!})$ with significant $\omega
_{Z_{_{\!}+}}^{(_{\!}I_{\!},\xi ,I_{\!+\!},\theta _{\!+\!})}$ in the space $%
\mathcal{F}(\mathbb{L}_{\!+})\times \Theta _{+}$, we:

\begin{enumerate}
\item seek $(I_{+}^{(1)},\theta _{\!+\!}^{(1)})$\ with significant $%
w_{Z_{_{\!}+}}^{(_{\!}1,I_{\!}^{(1)},\xi ^{(1)},I_{\!+\!}^{(1)},\theta
_{\!+\!}^{(1)})}$\ from the smaller space $\mathcal{F}(\mathbb{L}%
_{\!+}^{(1)})\times \Theta _{+}^{(1)}$;

\item for each such $(I_{+}^{(1)},\theta _{\!+\!}^{(1)})$ find the $%
(I_{+}^{(0)},\theta _{\!+\!}^{(0)})$ with the best $w_{Z_{_{\!}+}}^{(_{%
\!}0,I_{\!}^{(0)},\xi ^{(0)},I_{\!+\!}^{(0)},\theta _{\!+\!}^{(0)})}$,
subject to the constraint 
\begin{equation}
1_{\mathbb{L}_{\!+}^{(1)}}\theta _{\!+\!}^{(1)}+1_{\mathbb{L}%
_{\!+}^{(0)}}\theta _{\!+\!}^{(0)}\in {\Theta }(I_{+}^{(1)}\uplus
I_{+}^{(0)});  \label{eq:constraint}
\end{equation}

\item construct $(I_{\!+},\theta _{\!+\!})$ from $(I_{+}^{(1)},\theta
_{\!+\!}^{(1)})$ and $(I_{+}^{(0)},\theta _{\!+\!}^{(0)})$ via (\ref%
{eq:composition}) and compute the corresponding weight via (\ref%
{eq:componentweight}).
\end{enumerate}

Due to the constraint \ref{eq:constraint}, $1_{{\Theta }(I_{\!+\!})}(\theta _{\!+\!})=1$, and
hence, it follows from (\ref{eq:componentweight}) that the resulting GLMB
component $(I_{\!},\xi ,I_{\!+\!},\theta _{\!+\!})$ also has significant
weight.

The advantage of this strategy is two fold:

\begin{itemize}
\item searching over a much smaller space $\mathcal{F}(\mathbb{L}%
_{\!+}^{(1)})\times \Theta _{+}^{(1)}$ results in a linear complexity in the
measurements $\mathcal{O}(T(P^{(1)})^{2}M)$ since typically $P^{(1)}<<M$; 

\item finding $(I_{+}^{(0)},\theta _{\!+\!}^{(0)})$ with the best weight
subject to the constraint $\theta _{\!+\!}\in {\Theta }(I_{\!+\!})$ is
straight forward and requires miminal computation.
\end{itemize}
\vspace{-3mm}
\subsection{Propagating Objects of Interest}

One way to generate significant $(I_{+}^{(1)},\theta _{\!+\!}^{(1)})$\ is to
design a Gibbs sampler with stationary distribution $w_{Z_{_{\!}+}}^{(_{%
\!}1,I_{\!}^{(1)},\xi ^{(1)},I_{\!+\!}^{(1)},\theta _{\!+\!}^{(1)})}$.
However, this approach requires computing the hybrid multi-object density,
which we try to avoid in the first place.

A much more efficient alternative is to treat the multi-Bernoulli clutter as
Poisson with matching intensity, and apply the standard GLMB filter (the
JMS-GLMB filter (\ref{eq:JMSGLMB}) with a single-mode), where the Gibbs
sampler \cite{Murty68} (or Murty's algorithm \cite{Cassella_Gibbs92}) can be
used to obtain significant $(I_{+}^{(1)},\theta _{\!+\!}^{(1)})$ \cite%
{VVH_GLMB17}. Since there are $|I^{(0)}| $ clutter generators from the
previous time with survival probability $P_{S}^{(0)}$, and $|\mathbf{\mathbb{%
B}}_{+}^{(0)}|$ clutter birth with probability $r_{B,+}^{(0)}$, the
predicted clutter intensity is given by $\hat{\kappa}%
_{+}=(P_{S}^{(0)}|I^{(0)}|+r_{B,+}^{(0)}|\mathbf{\mathbb{B}}%
_{+}^{(0)}|)P_{D,+}^{(0)}V^{-1}$. Note that a Poisson RFS has larger
variance on the number of clutter points than a multi-Bernoulli with
matching intensity. Hence, in treating clutter as a Poisson RFS, we are
effectively tempering with the clutter model to induce the Gibbs sampler (or
Murty's algorithm) to generate more diverse components \cite{VVH_GLMB17}.

Following \cite{VVH_GLMB17}, let us enumerate $Z_{+}=\{z_{1:M}\}$, $%
I^{(1)}=\{\ell _{1:R}\}$, and $\mathbb{B}_{+}^{(1)}=\{\ell _{R+1:P}\}$. The $%
(I_{\!+}^{(1)},\theta _{\!+\!}^{(1)})\in $ $\mathcal{F}(\mathbb{L}%
_{+}^{(1)})\times \Theta (I_{+}^{(1)})$ at time $k+1$ with significant
weights are determined by solving a ranked assignment problem with cost
matrix $[\eta _{i}^{(\xi ^{(1)})}(j)]$, $i=1:P$,$~j=-1:M$, where%
\begin{equation*}
\eta _{i}^{(\xi ^{(1)})\!}(j)\!=\!\left\{ \!\!%
\begin{array}{ll}
1-\bar{P}_{S\!}^{(\xi ^{(1)})}(1,\ell _{i}) & \!\!\ell _{i}\in I^{(1)},~j<0
\\ 
\bar{P}_{S}^{(\xi ^{(1)})}(1,\ell _{i})\bar{\psi}_{Z_{+}}^{(\xi
^{(1)},\theta _{+}^{(1)})\!}(1,\ell _{i}) & \!\!\ell _{i}\in I^{(1)},~j\geq 0
\\ 
1-r_{B,+}(\ell _{i}) & \!\!\ell _{i}\in \mathbb{B}_{+}^{(1)},~j<0 \\ 
r_{B,+}(\ell _{i})\bar{\psi}_{Z_{+}}^{(\xi ^{(1)},\theta
_{+}^{(1)})\!}(1,\ell _{i}) & \!\!\ell _{i}\in \mathbb{B}_{+}^{(1)},~j\geq 0%
\end{array}%
\right.
\end{equation*}%
\begin{eqnarray*}
\bar{\psi}_{Z_{+}}^{(\xi ^{(1)},\theta _{+}^{(1)})}(1,\ell )\!\!
&=&\!\!\left\langle \bar{p}_{+}^{(\xi ^{(1)})}(\cdot ,1,\ell ),\psi
_{Z_{+}}^{(\theta _{+}^{(1)}(\ell ))}(\cdot ,1,\ell )\right\rangle \\
\psi _{Z_{+}}^{(j)}(\zeta ,1,\ell )\!\! &=&\!\!\left\{ \!%
\begin{array}{ll}
\frac{P_{D,+}^{(1)}(\zeta ,\ell )g_{+}^{(1)}(z_{j}|\zeta ,\ell )\!}{\hat{%
\kappa}_{+}}, & \!\!\text{if }j\in \left\{ 1,...,M\right\} \\ 
1-P_{D,+}^{(1)}(\zeta ,\ell ), & \!\!\text{if }j=0%
\end{array}%
\right.
\end{eqnarray*}%
Such a ranked assignment problem can be solved by Murty's algorithm or the
Gibbs sampler given in Section III-D \cite{VVH_GLMB17}.
\vspace{-3mm}
\subsection{Propagating Clutter Generators}

Given $(I_{\!+}^{(1)},\theta _{\!+\!}^{(1)})$ pertaining to the objects of
interest, we proceed to determine $(I_{\!+}^{(0)},\theta _{\!+\!}^{(0)})$
pertaining to clutter generators, which maximizes $\omega
_{Z_{_{\!}+}}^{(_{\!}0,I_{\!}^{(0)},\xi ^{(0)},I_{\!+\!}^{(0)},\theta
_{\!+\!}^{(0)})}$ where $I_{\!+\!}^{(0)}\subseteq $ $_{\!}I_{\!}^{(0)}\cup 
\mathbb{B}_{+}^{(0)}$ and $\theta _{\!+\!}^{(0)}:I_{\!+\!}^{(0)}\rightarrow
\{0:M\}$ subject to constraint (\ref{eq:constraint}).

Denote by $Z_{+}^{(1)}\subseteq Z_{+}$ the set of measurements assigned to $%
I_{\!+}^{(1)}$ by $\theta _{\!+\!}^{(1)}$ and the remaining set of
measurements $Z_{+}-$ $Z_{+}^{(1)}$, due to clutter generators, by $%
Z_{+}^{(0)}$. Recall that clutter generators are functionally identical
except in label and that their propagation reduces to calculating their
corresponding weights (\ref{eq:CGprop}). Let $N_{S}^{(0)}=|I^{(0)}\cap
I_{+}^{(0)}|$ and $N_{B,+}^{(0)}=|\mathbb{B}_{{+}}^{(0)}\cap I_{+}^{(0)}|$
denote the counts of surviving and new born clutter generators respectively.
Then $|I^{(0)}-I_{+}^{(0)}|=|I^{(0)}|-N_{S}^{(0)}$ and $|\mathbb{B}%
_{+}^{(0)}-I_{+}^{(0)}|=|\mathbb{B}_{+}^{(0)}|-N_{B,+}^{(0)}$. Observe that
the count $|Z_{+}^{(0)}|$ of clutter must equal the number of detections of
clutter generators according to $(I_{\!+}^{(0)},\theta _{\!+\!}^{(0)})$,
i.e. $|Z_{+}^{(0)}|=|\{\ell \in I_{+}^{(0)}:\theta _{+}^{(0)}(\ell )>0\}|$
and hence the count of misdetections of clutter generators according to $%
(I_{\!+}^{(0)},\theta _{\!+\!}^{(0)})$ is $%
N_{S}^{(0)}+N_{B,+}^{(0)}-|Z_{+}^{(0)}|=|\{\ell \in I_{+}^{(0)}:\theta
_{+}^{(0)}(\ell )=0\}|$. Consequently the weight (\ref{eq:CGprop}) can be
rewritten as 
\begin{align*}
& \!\!\!\!\!\!\!\!\omega _{Z_{_{\!}+}}^{(0,I_{\!}^{(0)},\xi
^{(0)},I_{\!+\!}^{(0)},\theta _{\!+\!}^{(0)})}\! \\
=& \left[ 1-P_{S}^{(0)\!}\right] ^{\!|I^{(0)\!}|-N_{S}^{(0)}}\!\!\!\left[
P_{S}^{(0)\!}\right] ^{\!N_{S}^{(0)}}\!\!\!\left[ 1-r_{B,+\!}^{(0)}\right]
^{\!|\mathbb{B}_{+}^{(0)}|-N_{B,+}^{(0)}}\times \\
& \left[ r_{B,+\!}^{(0)}\right] ^{\!N_{B,+}^{(0)}}\left[ 1-P_{D,+}^{(0)\!}%
\right] ^{\!N_{S}^{(0)}+N_{B,+}^{(0)}-|Z_{+}^{(0)\!}|}\!\left[
P_{D,+}^{(0)\!}V^{-1}\right] ^{\!|Z_{+}^{(0)}|} \\
\propto & \left[ \!\frac{P_{S}^{(0)}(1-P_{D,+}^{(0)})}{1-P_{S}^{(0)}}\!%
\right] ^{\!N_{S}^{(0)}}\!\left[ \!\frac{r_{B,+}^{(0)}(1-P_{D,+}^{(0)})}{%
1-r_{B,+}^{(0)}}\!\right] ^{\!N_{B,+}^{(0)}}
\end{align*}%
Thus seeking the best $(I_{\!+}^{(0)},\theta _{\!+\!}^{(0)})$ subject to
constraint (\ref{eq:constraint}) reduces to seeking the best $%
(N_{S}^{(0)},N_{B,+}^{(0)})$ subject to the constraints $0\leq
N_{S}^{(0)}\leq $ $|I^{(0)}|$, $0\leq N_{B,+}^{(0)}\leq $ $|\mathbb{B}%
_{+}^{(0)}|$ and $N_{S}^{(0)}+N_{B,+}^{(0)}\leq |Z_{+}^{(0)}|$.
\vspace{-3mm}
\subsection{Linear Gaussian Update Parameters}

Let $\mathcal{N}(\cdot ;\bar{\zeta},P)$ denotes a Gaussian density with mean 
$\bar{\zeta}$ and covariance $P$. Then for a linear Gaussian multi-object
model of the objects of interest $P_{S}^{(1)}(\zeta ,\ell )=P_{S}^{(1)}$, $%
P_{D}^{(1)}(\zeta ,\ell )=P_{D}^{(1)}$, $f_{+}^{(1)}(\zeta _{+}|\zeta ,\ell
)=\mathcal{N}(\zeta _{+};F\zeta ,Q)$, $g^{(1)}(z|\zeta ,\ell )=\mathcal{N}%
(z;H\zeta ,R)$, and $p_{B,+}^{(1)}(\zeta _{+})=\mathcal{N}(\zeta _{+};\bar{%
\zeta}_{+}^{(1)},P_{+}^{(1)})$, where $F$ is the transition matrix, $Q$ is
the process noise covariance, $H$ is the observation matrix, $R$ is the
observation noise covariance, $\bar{\zeta}_{+}^{(1)}$ and $P_{+}^{(1)}$ are
the mean and covariance of the kinematic state of a new object of interest.
If each current density of an object of interest is a Gaussian of the form%
\begin{equation}
p^{(\xi ^{(1)})_{\!}}(\zeta ,1,\ell )=\mathcal{N}(\zeta ;\bar{\zeta}^{(\xi
^{(1)})}(\ell ),P^{(\xi ^{(1)})}(\ell ))  \label{eq:GM_single_pdf}
\end{equation}%
then the terms (\ref{eq:mclassGLMB4}), (\ref{eq:mclassGLMB5}), (\ref%
{eq:mclassGLMB6}) can be computed analytically using the following
identities: 
\begin{equation*}
\int \mathcal{N}(\zeta ;\bar{\zeta},P)\mathcal{N}(\zeta _{+};F\zeta
,Q)d\zeta =\mathcal{N}(\zeta _{+};F\bar{\zeta},FPF^{T}+Q),
\end{equation*}%
\begin{eqnarray*}
&&\!\!\!\!\!\!\!\!\!\!\!\!\!\!\!\!\!\!\!\!\!\!\!\!\!\!\!\!\!\!\mathcal{N}%
(\zeta ;\bar{\zeta},P)\mathcal{N}(z;H\zeta ,R) \\
&=&q(z)\mathcal{N}(\zeta ;\bar{\zeta}+K(z-H\bar{\zeta}),[I-KH]P), \\
q(z) &=&\mathcal{N}(z;H\bar{\zeta},HPH^{T}+R), \\
K &=&PH^{T}\left[ HPH^{T}+R\right] ^{-1}.
\end{eqnarray*}
\vspace{-4mm}
\subsection{Extension to Unknown Detection Probability}

\label{subsec_UPD}

Following the approach in \cite{MVV10}, to jointly estimate an unknown
detection probability, we augment a variable $a\in \lbrack 0,1]$ to the
state, i.e. $\mathbf{x}=(\zeta ,m,a,\ell )$, so that 
\begin{equation}
P_{D}^{(m)}(\zeta ,a,\ell )=a.
\end{equation}%
Additionally, in this model $g^{(m)}(z|\zeta ,a,\ell )=g^{(m)}(z|\zeta ,\ell
)$, $P_{S}^{(m)}(\zeta ,a,\ell )=P_{S}^{(m)}$, $p_{B,+}^{(1)}(\zeta
_{+},a_{+})=p_{B,+}^{(1)}(\zeta _{+})p_{B,+}^{(1)}(a_{+})$, and the
transition density is given by 
\begin{equation}
f_{+}^{(m)}(\zeta _{+},a_{+}|\zeta ,a,\ell )=f_{+}^{(m)}(\zeta _{+},|\zeta
,\ell )f_{+}^{(\Delta )}(a_{+}|a).
\end{equation}%
The unknown detection probability is then modelled on a Beta distribution $%
\beta (\cdot ,s,t)$ where $s$ and $t$ are positive shape parameters and the
single-object state density is modelled by a Beta-Gaussian density:%
\begin{eqnarray*}
&&\!\!\!\!\!\!\!p^{(\xi ^{(1)})_{\!}}(\zeta ,1,a,\ell ) \\
&=&\beta (a;s^{(\xi ^{(1)})}(\ell ),t^{(\xi ^{(1)})}(\ell ))\mathcal{N}%
(\zeta ;m^{(\xi ^{(1)})}(\ell ),P^{(\xi ^{(1)})}(\ell )
\end{eqnarray*}

Note that in practice, we only use the Beta model for the unknown detection
probability of the objects of interest. For clutter generators, we use a
fixed detection probability between 0.5 and 1. Values close to 0.5 result in
a large variance on the clutter cardinality and faster reponse to changes in
clutter parameter, while the converse is true for values close to 1.

Analytic computation of the terms (\ref{eq:mclassGLMB4}), (\ref%
{eq:mclassGLMB5}), (\ref{eq:mclassGLMB6}) can be performed separately for
the Gaussian part (which has been given in the previous subsection) and the
Beta part using \cite{MVV10}: 
\begin{equation*}
\beta (a_{+};s_{+},t_{+})=\int \beta (a;s,t)f_{+}^{(\Delta )}(a_{+}|a)da
\end{equation*}%
where 
\begin{eqnarray*}
s_{+} &=&\left( \frac{\mu _{\beta }(1-\mu _{\beta })}{\sigma _{\beta }^{2}}%
-1\right) \mu _{\beta }, \\
t_{+} &=&\left( \frac{\mu _{\beta }(1-\mu _{\beta })}{\sigma _{\beta }^{2}}%
-1\right) \left( 1-\mu _{\beta }\right) . \\
\mu _{\beta } &=&\frac{s}{s+t}\;,\;\;\;\;\;\sigma _{\beta }^{2}=\frac{st}{%
(s+t)^{2}(s+t+1)}
\end{eqnarray*}%
(note that $\beta (\cdot ;s_{+},t_{+})$ has the same mean $\mu _{\beta }$ as 
$\beta (\cdot ;s,t)$ but a larger variance than $\sigma _{\beta }$) and%
\begin{eqnarray*}
(1-a)\beta (a;s,t) &=&\frac{B(s,t+1)}{B(s,t)}\beta (a;s,t+1), \\
a\beta (a;s,t) &=&\frac{B(s+1,t)}{B(s,t)}\beta (a;s+1,t),
\end{eqnarray*}%
where $B(s,t)=\tint\nolimits_{0}^{1}a^{s-1}(1-a)^{t-1}da$.

\section{Numerical Studies}

\subsection{Simulations}

The following simulation scenario is used to test the proposed robust
multi-object filter. The target state vector $[x,y,\dot{x},\dot{y}]^{T}$
consists of cartesian coordinates and the velocities. Objects of interest
move according to a constant velocity model, with zero-mean Gaussian process
noise of covariance%
\begin{equation*}
Q_{f}={v_{f}}^{2}%
\begin{bmatrix}
T^{4}/4 & T^{3}/2 & 0 & 0 \\ 
T^{3}/2 & T^{2} & 0 & 0 \\ 
0 & 0 & T^{4}/4 & T^{3}/2 \\ 
0 & 0 & T^{3}/2 & T^{2}%
\end{bmatrix}%
\end{equation*}%
where $v_{f}=5ms^{-1}$ and $T=1s$. Objects of interest are born from a
labeled multi Bernoulli distribution with four components of 0.03 birth
probability, and birth densities 
\begin{eqnarray*}
&&\mathcal{N}(\cdotp,[0,0,0,0]^{T},P_{\gamma }), \\
&&\mathcal{N}(\cdotp,[400,-600,0,0]^{T},P_{\gamma }), \\
&&\mathcal{N}(\cdotp,[-800,-200,0,0]^{T},P_{\gamma }), \\
&&\mathcal{N}(\cdotp,[-200,800,0,0]^{T},P_{\gamma }),
\end{eqnarray*}%
where $P_{\gamma }=\diag([50,50,50,50]).$ The probability of survival is set
at 0.99.

Objects of interest enter and leave the observation region $%
[-1000,1000]m\times \lbrack -1000,1000]m$ at different times reaching a maximum
of ten targets. The measurements are the object positions
obtained through a sensor located at coordinate $(0,0)$. Measurement noise
is assumed to be distributed Gaussian with zero-mean and covariance $Q_{r}
$ where $v_{r}=3ms^{-1}$. 
\begin{equation*}
Q_{r}={v_{r}}^2%
\begin{bmatrix}
1 & 0 \\ 
0 & 1%
\end{bmatrix}%
\end{equation*}

The detection model parameters for all new born objects of interest are set
at $s= 9$ and $t=1$ resulting in a mean of 0.9 for the detection probability.
At the initial timestep, clutter generators are born from a (labeled)
multi-Bernoulli distribution with 120 components, each with 0.5 birth
probability and uniform birth density. At subsequent timesteps clutter
generators are born from a (labeled) multi-Bernoulli distribution with 30
components, each with 0.5 birth probability and uniform birth density.
Probability of survival and probability of detection of the clutter
generators are both set at 0.9.

Four scenarios corresponding to four different pairings of average
(unknown) clutter rate and detection probability (see Table 1) are studied. 
% Please add the following required packages to your document preamble:

\begin{table}[tbp]
\centering
\begin{tabular}{|c|c|c|}
\hline
\multirow{2}{*}{Scenario ID} & \multirow{2}{*}{Clutter Rate} & %
\multirow{2}{*}{Detection Probability} \\  
&  &  \\ 
\hline
\multirow{2}{*}{1} & \multirow{2}{*}{10} & \multirow{2}{*}{0.97} \\ [2ex]
\hline
%&  &  \\ \hline
\multirow{2}{*}{2} & \multirow{2}{*}{10} & \multirow{2}{*}{0.85} \\ [2ex]
\hline
%&  &  \\ \hline
\multirow{2}{*}{3} & \multirow{2}{*}{70} & \multirow{2}{*}{0.97} \\ [2ex]
\hline
%&  &  \\ \hline
\multirow{2}{*}{4} & \multirow{2}{*}{varying between 25-35} & \multirow{2}{*}{0.95} \\ [2ex]
\hline
\end{tabular}
\vspace{\baselineskip}
\caption{Simulation Parameters unknown to the filter}
\label{my-label}
\end{table}

The Fig. \ref{fig:fig1subfig1} shows the OSPA\cite{SVV08} errors obtained from 100 Monte Carlo
runs (OSPA c = 300, p = 1) for the proposed GLMB filter in comparison with $%
\lambda $-CPHD\cite{MVV10} filter for scenario 1.
Estimated clutter rates and detection probabilites by the two filters are
shown in Fig.  \ref{fig:fig1subfig2}, while estimated tracks for objects of interest taken from
a single run is shown in Fig. \ref{fig:fig1subfig3}. It can be seen that for the given
parameters, the GLMB filter performs far better than the $\lambda $-CPHD 
 in terms of clutter rate, detection probability and track
estimation for objects of interest.

We further investigate the performance of the proposed algorithm by varying
the background parameters in scenarios 2 and 3. The average detection
probability in scenario 2 is lower than that of scenario 1, while the
average clutter rate in scenario 3 is higher than that of scenario 1.  Note from Figure \ref{fig:Scenario1} that 
$\lambda $-CPHD filter begins to fail in scenario 1. The
OSPA errors for 100 Monte Carlo runs, estimates of the clutter
rate and detection probabilities for the more challenging scenarios 2 and 3 are given in Fig.
\ref{fig:Scenario2}, Fig. \ref{fig:Scenario3}  at which $\lambda $-CPHD competely breaks down. On the other hand the proposed GLMB filter is capable of accurately tracking the objects of
interest as well as estimating the unknown clutter and detection parameters. 
The fourth scenario comprises of a wavering  clutter rate with comparison to the 
$\lambda $-CPHD\ filter.  Perceiving Fig. \ref{fig:Scenario4} it is clear that that the proposed filter outperforms $\lambda $-CPHD and is quite adept at converging swiftly to the shifted clutter rate.

\begin{figure}[ht]
\centering
\subfigure[OSPA Error]{   
    \label{fig:fig1subfig1}     
    \includegraphics[trim=20 0 30 0,clip,width=0.48\textwidth]{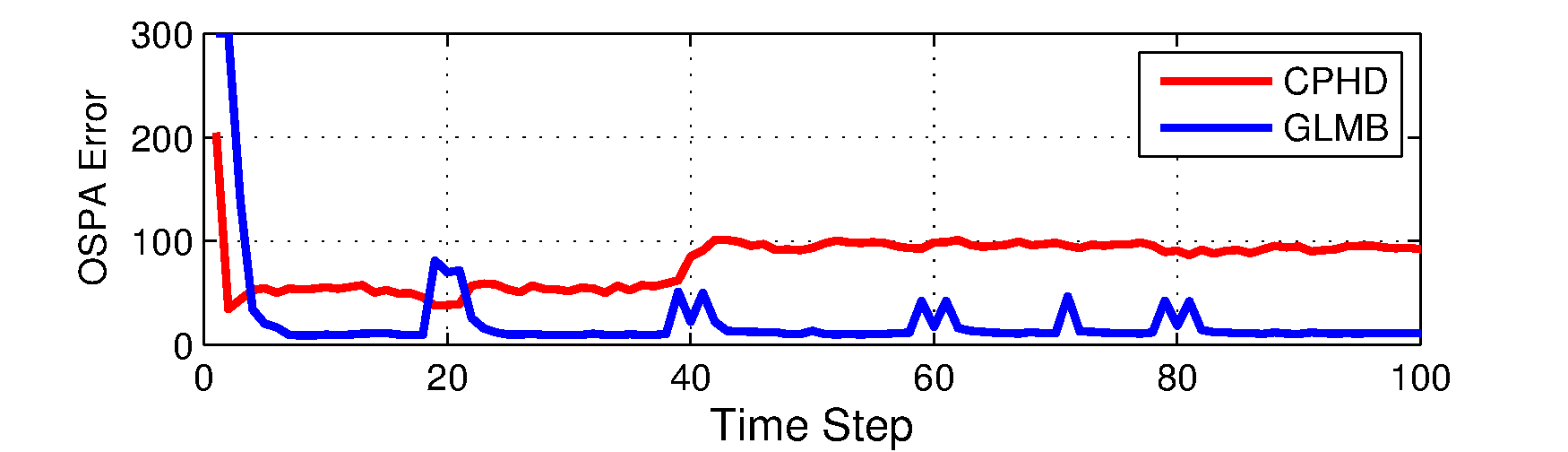}	
}
\subfigure[estimated clutter and detection parameters]{    
    \label{fig:fig1subfig2}
    \includegraphics[trim=30 0 20 0,clip,width=0.5\textwidth]{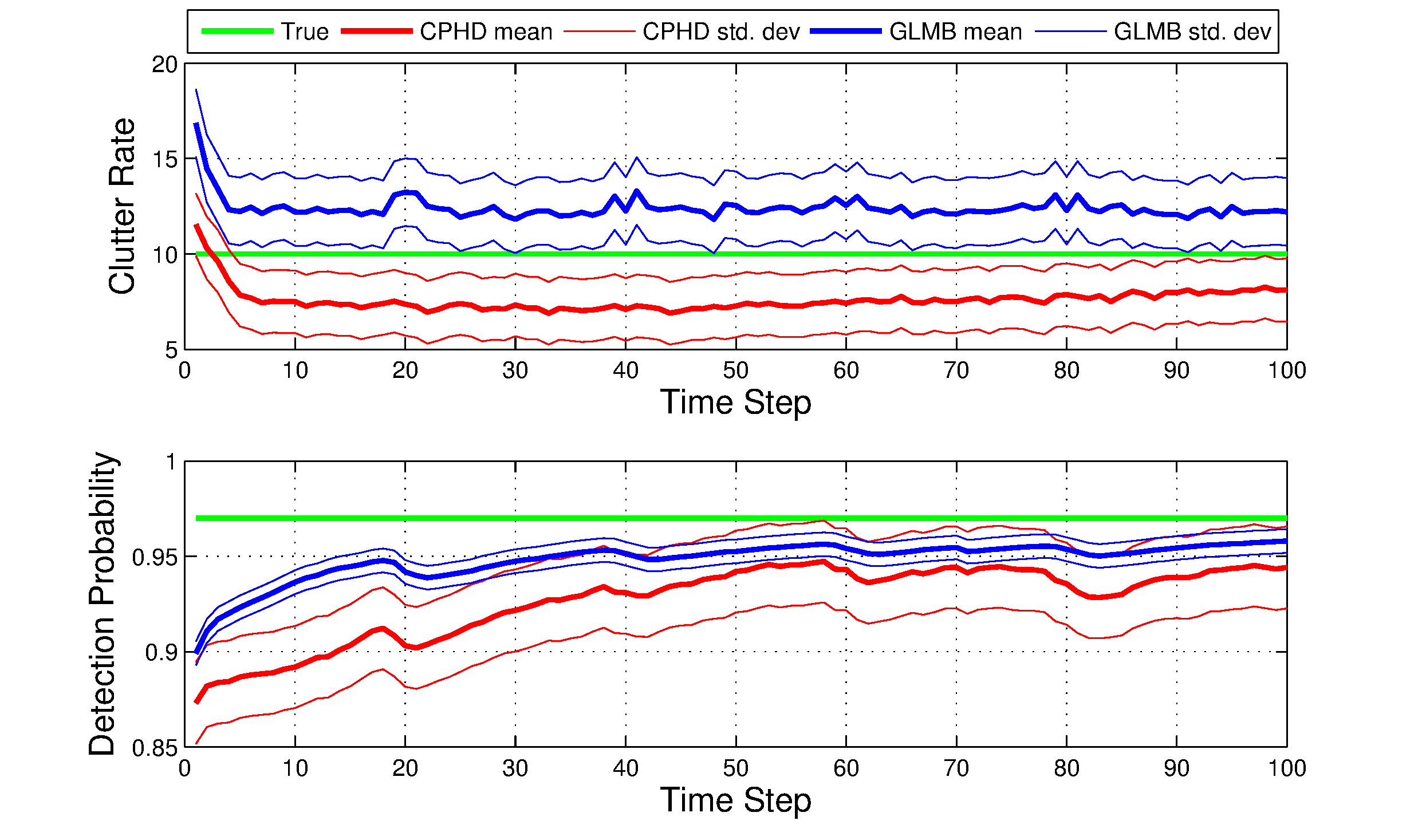}	
}
\subfigure[Track Estimations]{    
    \label{fig:fig1subfig3}
    \includegraphics[trim=10 0 10 0,clip,width=0.5\textwidth]{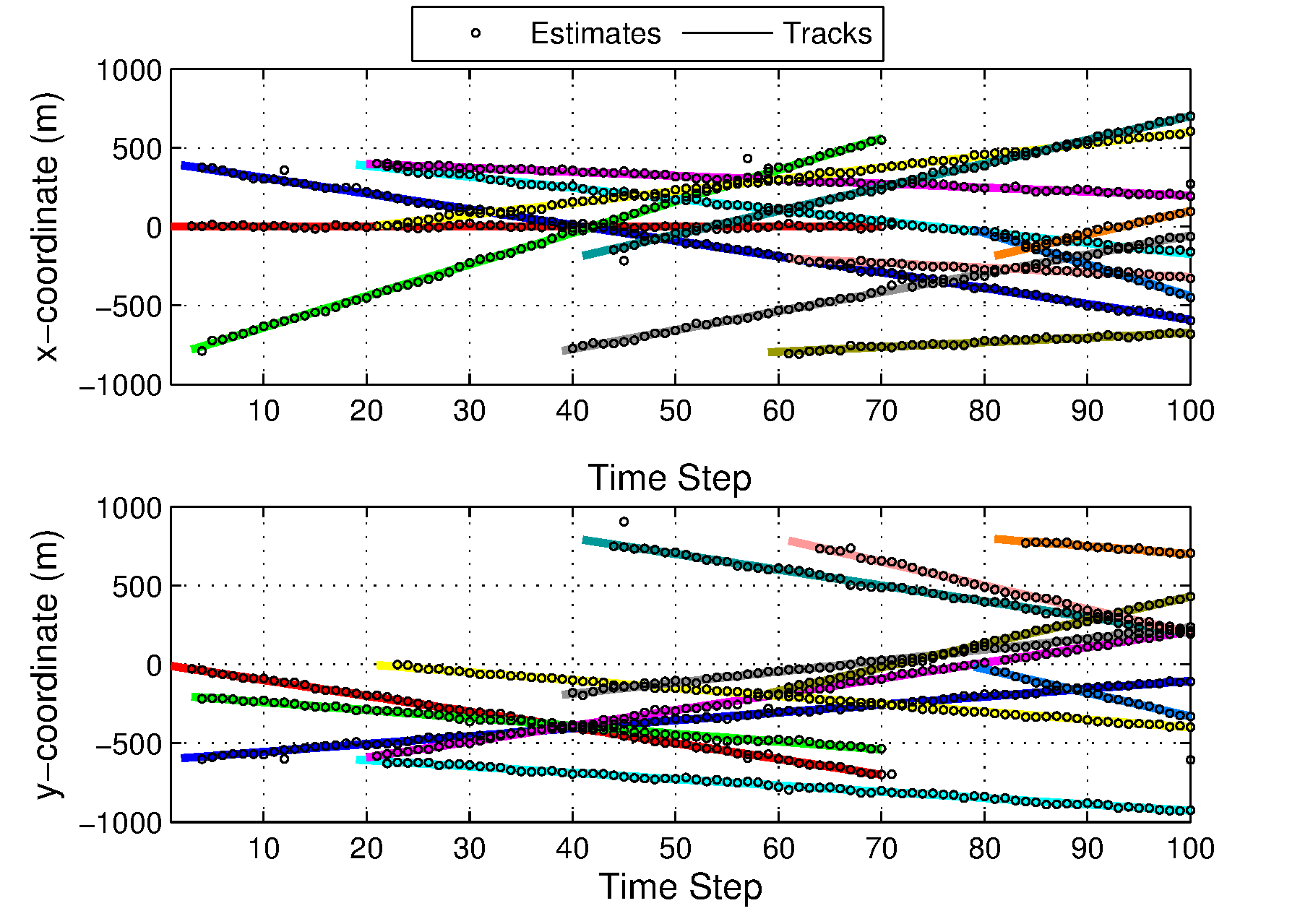}	
}
\caption[Scenario 1]{Scenario 1. The bumps in the OSPA error for GLMB in \ref{fig:fig1subfig1} appear close to time steps where a new birth or a death of an object of interest occurs. }
\label{fig:Scenario1}
\end{figure}
%----------------------------------------------------------------------------
\begin{figure}[ht]
\centering
\subfigure[OSPA Error]{   
    \label{fig:fig2subfig1}     
    \includegraphics[trim=15 0 30 0,clip,width=0.48\textwidth]{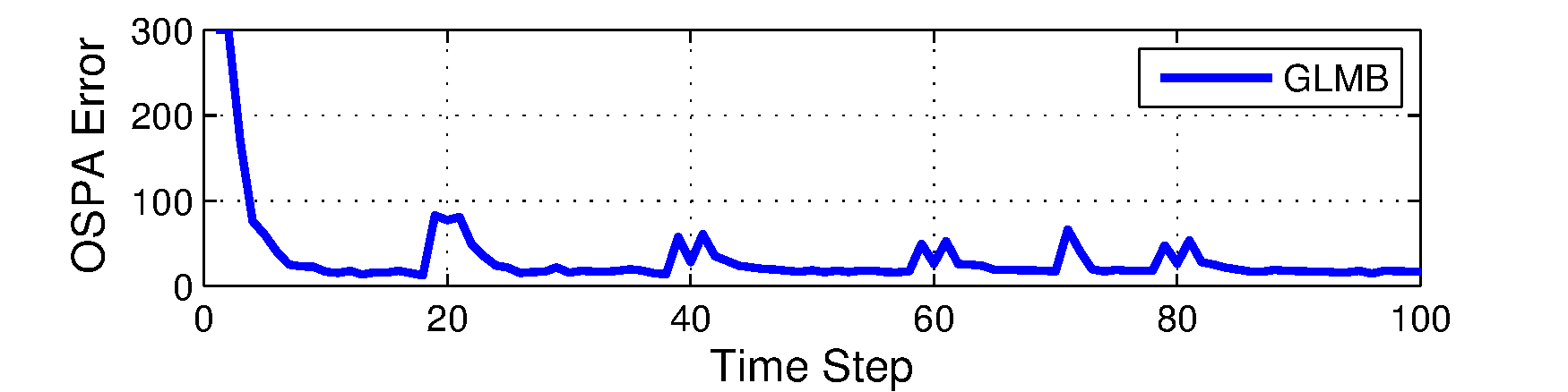}		
}
\subfigure[estimated clutter and detection parameters]{    
    \label{fig:fig2subfig2}
    \includegraphics[trim=10 0 20 0,clip,width=0.5\textwidth]{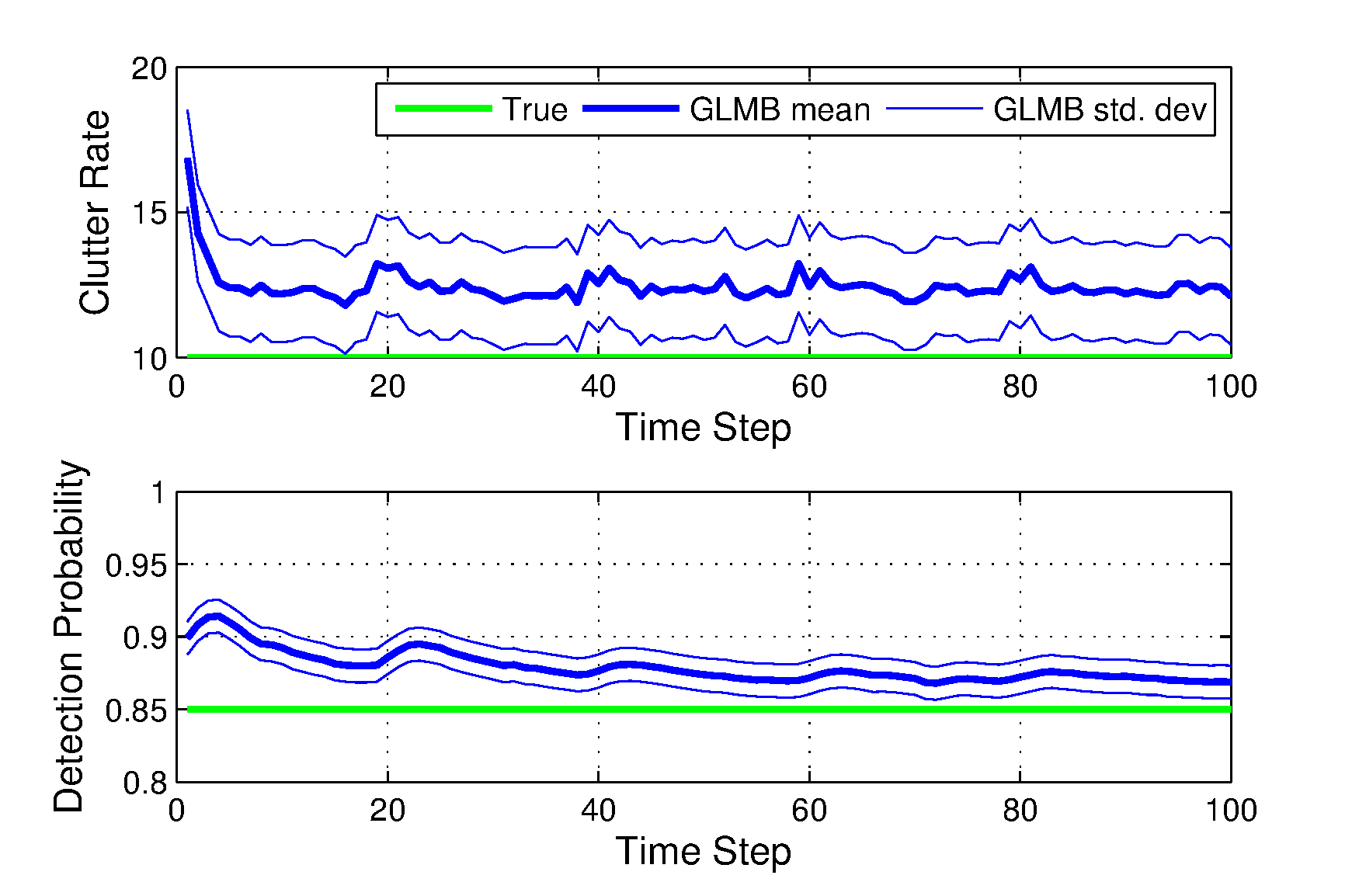}	
}
\caption[Scenario 2]{Scenario 2. Comparison with $\lambda $-CPHD not included as it completely fails at this juncture.}
\label{fig:Scenario2}
\end{figure}
%----------------------------------------------------------------------------
\begin{figure}[ht]
\centering
\subfigure[OSPA Error]{   
    \label{fig:fig3subfig1}     
    \includegraphics[trim=15 0 30 0,clip,width=0.48\textwidth]{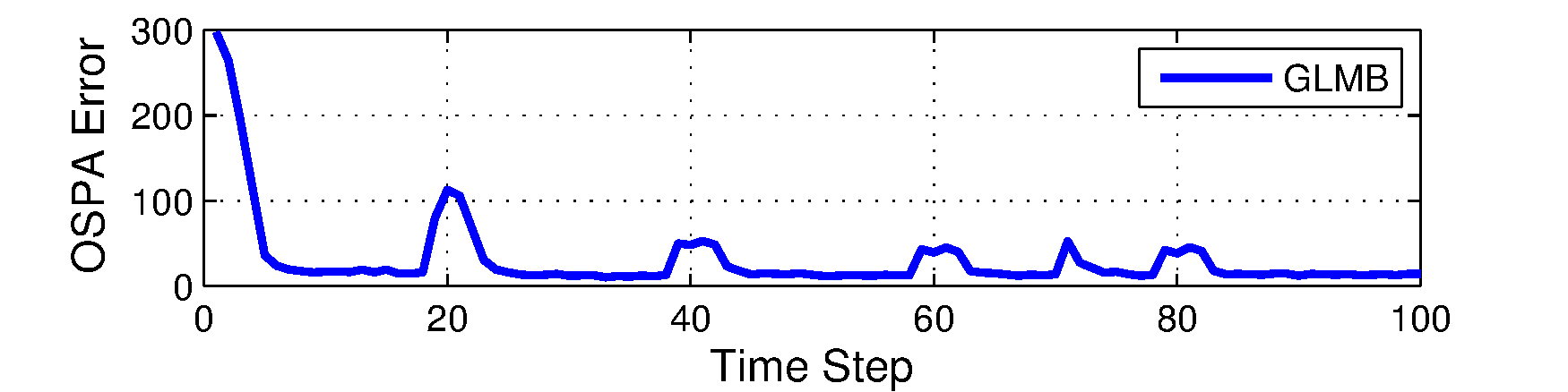}		
}
\subfigure[estimated clutter and detection parameters]{    
    \label{fig:fig3subfig2}
    \includegraphics[trim=10 0 20 0,clip,width=0.5\textwidth]{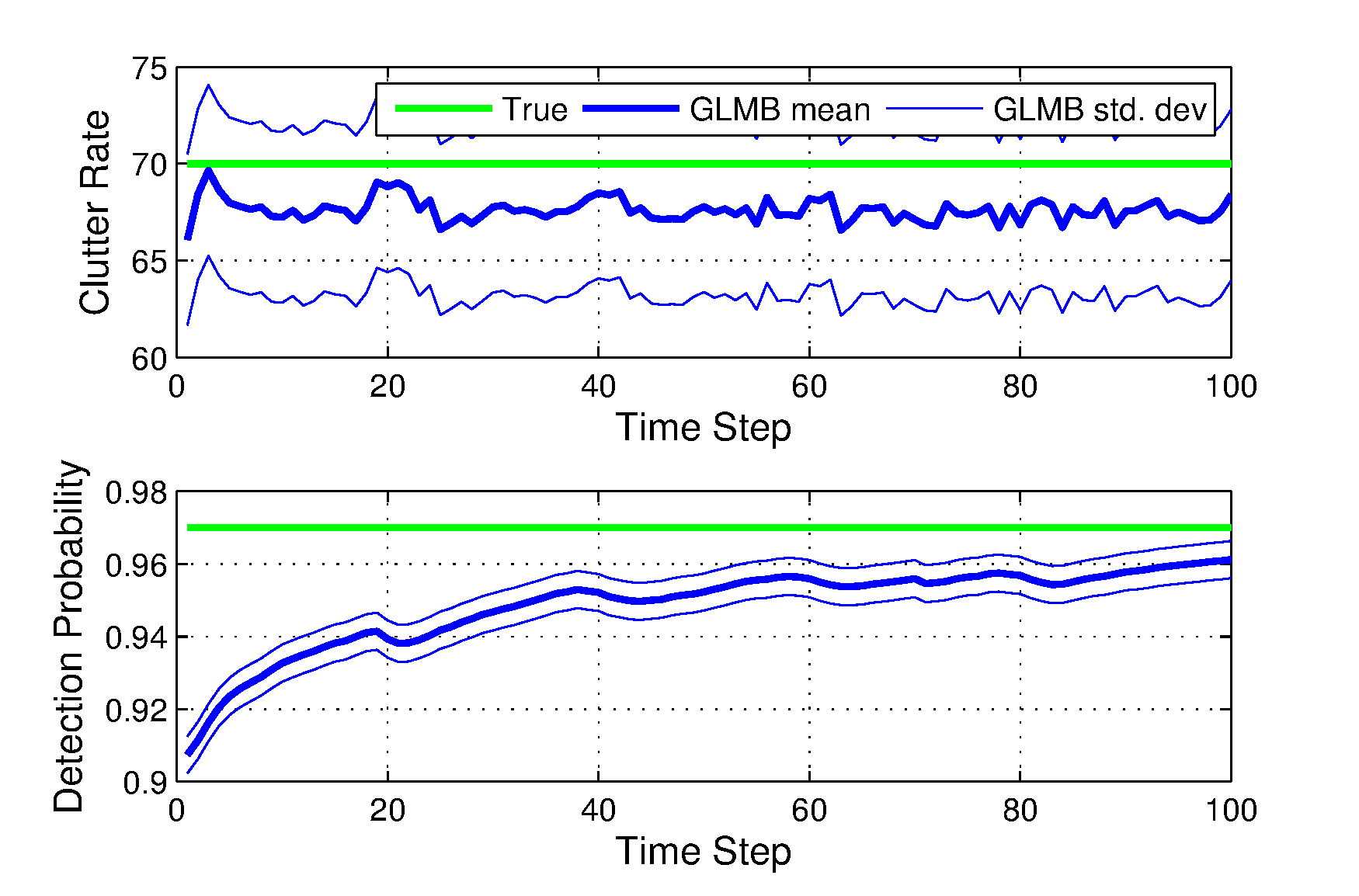}	
}
\caption[Scenario 3]{Scenario 3. Comparison with $\lambda $-CPHD not included as it completely fails at this juncture.}
\label{fig:Scenario3}
\end{figure}
%----------------------------------------------------------------------------
\begin{figure}[ht]
\centering
\subfigure[OSPA Error]{   
    \label{fig:fig4subfig1}     
    \includegraphics[trim=20 0 30 0,clip,width=0.48\textwidth]{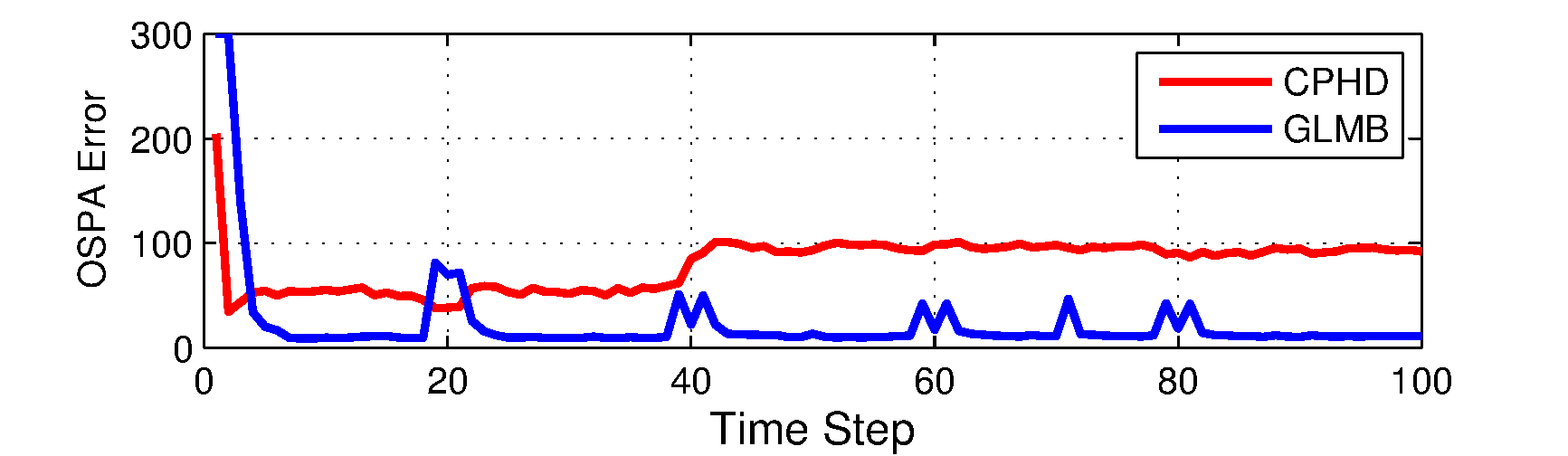}		
}
\subfigure[estimated clutter and detection parameters]{    
    \label{fig:fig4subfig2}
    \includegraphics[trim=20 10 20 0,clip,width=0.5\textwidth]{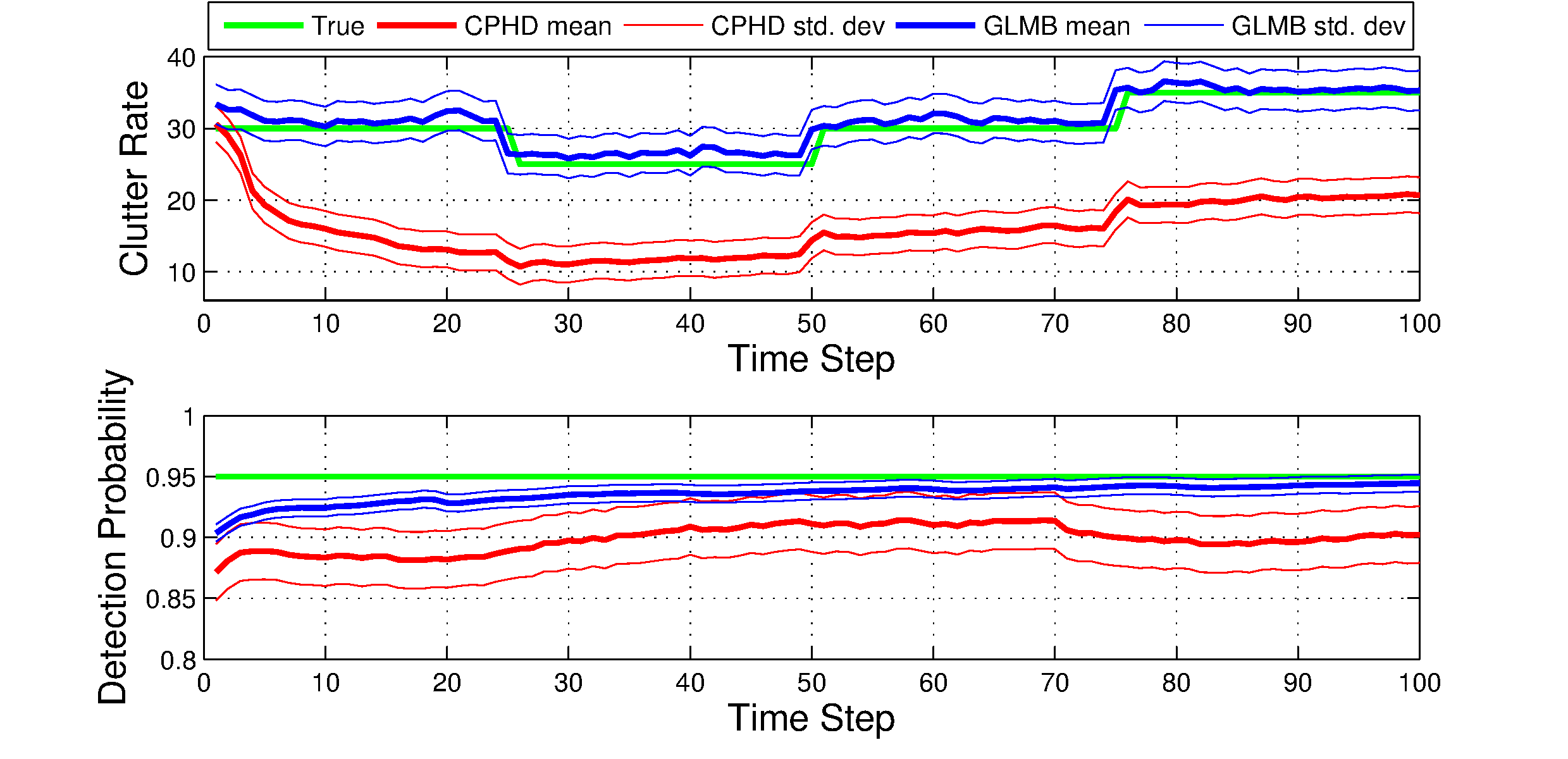}	
}
\caption[Scenario 4]{Scenario 4.}
\label{fig:Scenario4}
\end{figure}
\vspace{-4mm}
\subsection{Video Data}

The proposed filter for jointly unknown clutter rate and detection probability is tested on two image sequences: S2.L1 from PETS2009 datasets \cite{PETS} and KITTI-17 from KITTI datasets \cite{Geiger2012CVPR}. The detections are obtained using the detection algorithm in \cite{DOLLAR2014}.

\textit{Dataset 1:} The state vector consists of the target $x,y$ positions
and the velocities in each direction. The process noise is assumed to be
distributed from a zero-mean Gaussian with covariance $Q_{f}$ where $v_{f}=2$
pixels. Actual targets are assumed to be born from a labeled multi Bernoulli
distribution with seven components of 0.03 birth probability, and Gaussian
birth densities, 
\vspace{-2mm}
\begin{eqnarray*}
&&\!\!\!\!\!\!\mathcal{N}(\cdotp,[260;260;0;0]^{T},P_{\gamma }), 
\mathcal{N}(\cdotp,[740;370;0;0]^{T},P_{\gamma }), \\
&&\!\!\!\!\!\!\mathcal{N}(\cdotp,[10;200;0;0]^{T},P_{\gamma }), 
\mathcal{N}(\cdotp,[280;80;0;0]^{T},P_{\gamma }), \\
&&\!\!\!\!\!\!\mathcal{N}(\cdotp,[750;130;0;0]^{T},P_{\gamma }), 
\mathcal{N}(\cdotp,[650;270;0;0]^{T},P_{\gamma }), \\
&&\!\!\!\!\!\!\mathcal{N}(\cdotp,[500;200;0;0]^{T},P_{\gamma }), \text{ where } P_{\gamma }=\diag([10;10;3;3]).
\end{eqnarray*}%

The observation space is a $756\times 560$ pixel image frame. Actual target measurements contain the $x,y$
positions with measurement noise assumed to be distributed zero-mean Gaussian with covariance $Q_{r}$ with $v_{r}=3$ pixels. Clutter targets
are born from a multi Bernoulli distribution with 30 birth components in the
firstmost time step and 12 components in subsequent time steps each with 0.5
birth probability and uniform birth density. Probability of survival and
detection for clutter targets are both set at 0.9.

The Fig. \ref{fig:videotrack} shows tracking results at frames 20, 40 and 100
respectively. True and estimated clutter cardinality statistics are given in
Fig. \ref{fig:video1clrate}. From these figures it can be observed that the
filter successfully outputs object tracks and that the estimated clutter
rate nearly overlays the true clutter rate.

\textit{Dataset 2:} The detection results from this dataset (KITTI17) comprises of a higher number of false measurements than the PETS2009 S2.L1 dataset. The state vector consists of
the target $x,y$ positions and the velocities in each direction. The process
noise is assumed to be distributed from a zero-mean Gaussian with covariance 
$Q_{f}$ where $v_{f}=2$ pixels. Actual targets are assumed to be born from a
labeled multi Bernoulli distribution with three components of 0.05 birth
probability, and birth densities 
\begin{eqnarray*}
&&\!\!\!\!\!\!\mathcal{N}(\cdotp,[550;200;0;0]^{T},P_{\gamma }), 
\mathcal{N}(\cdotp,[1200;250;0;0]^{T},P_{\gamma }), \\
&&\!\!\!\!\!\!\mathcal{N}(\cdotp,[500;250;0;0]^{T},P_{\gamma }) \text{where } P_{\gamma }=\diag([10;10;1;1]).
\end{eqnarray*}
\vspace{-2mm}
State transition function for actual targets are based on constant velocity
model with a 0.99 probability of survival. Process noise is assumed to be
distributed from a zero-mean Gaussian with covariance $Q_{f}$ with $v_{f}=2$
pixels per frame. The observation space is a 1220$\times $350 pixel image
frame. Actual target measurements contain the $x,y$ positions with
measurement noise assumed to be distributed zero-mean Gaussian with
covariance $Q_{r}$ with $v_{r}=3$ pixels. Clutter target are born from 60
identical and uniformly distributed birth regions in the firstmost time step
and 20 birth regions in the subsequent time steps each with a birth
probability of 0.5. Probability of survival and detection for clutter
targets are both set at 0.9.

\begin{figure}[H]
\centering
\subfigure{    
    \label{fig:subfig1}
    \includegraphics[scale=0.4]{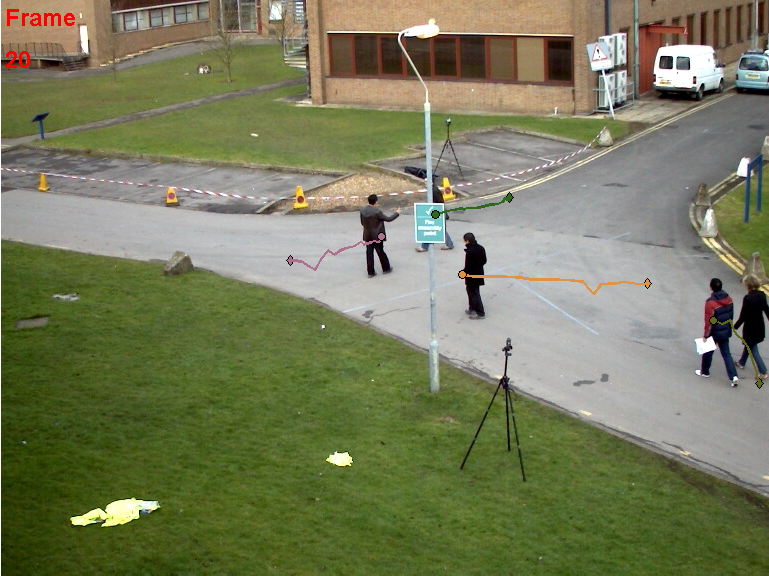}
} 
\subfigure{    
    \label{fig:subfig2}
    \includegraphics[scale=0.4]{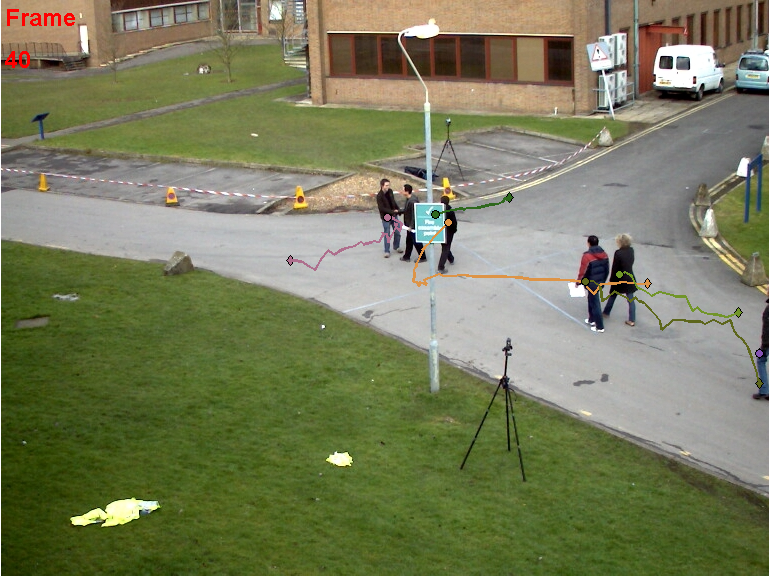}
} 
\subfigure{    
    \label{fig:subfig4}
    \includegraphics[scale=0.4]{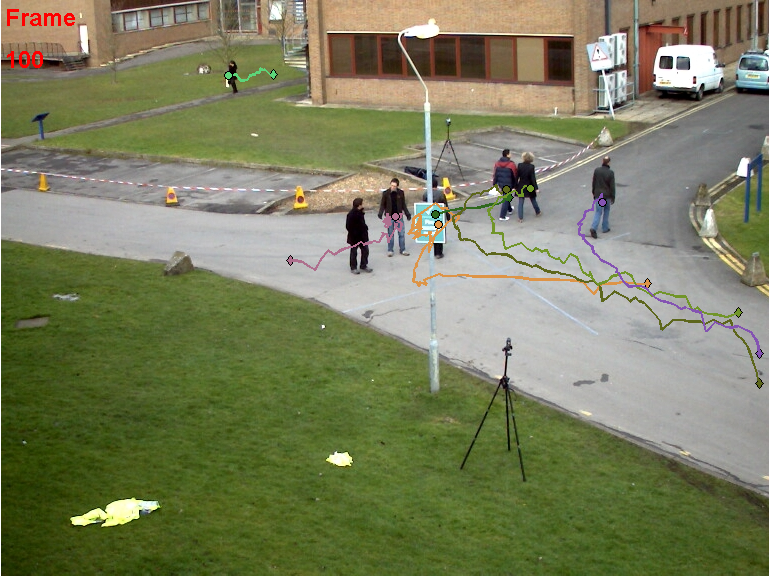}
}
\caption[Optional caption for list of figures]{Tracking results for frames
20, 40, 100 in dataset 1. }
\label{fig:videotrack}
\end{figure}

\begin{figure}[h]
\centering
\includegraphics[scale=0.55, trim= 35 0 20 20]{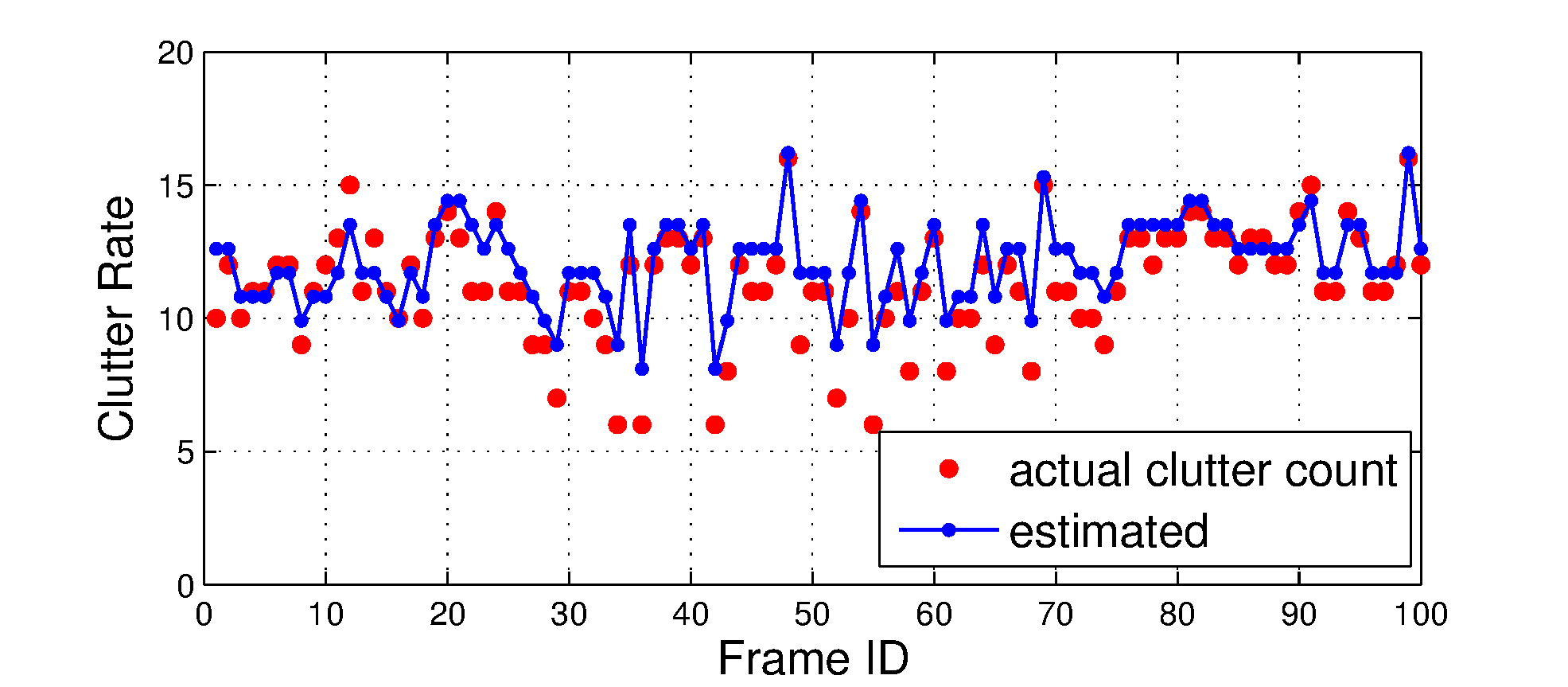}
\caption{Estimated clutter rate for dataset 1.}
\label{fig:video1clrate}
\end{figure}

The frames on the left of Fig. \ref{fig:videotrack2} shows tracking
results for frames 15, 35 and 50 obtained from the standard GLMB filter for
the guessed clutter rate of 60. The frames on the right of Fig. \ref%
{fig:videotrack2} shows tracking results for the same frames using the
proposed filter. When comparing each frame pair it can be noted that some
objects that were missed by the standard algorithm with the guessed clutter
rate has been picked up by the proposed algorithm. Comparison between true
and estimated clutter cardinality statistics given in Fig. \ref%
{fig:video2clrate} demonstrates that the estimated clutter rate is close
enough to the true clutter rate to achieve a similar performance if fed back to the standard algorithm\cite{VVP_GLMB13}.

\begin{figure*}[h]
\centering
\subfigure{
        \label{fig:subfig1}
    \includegraphics[scale=0.28, trim= 20 0 20 0]{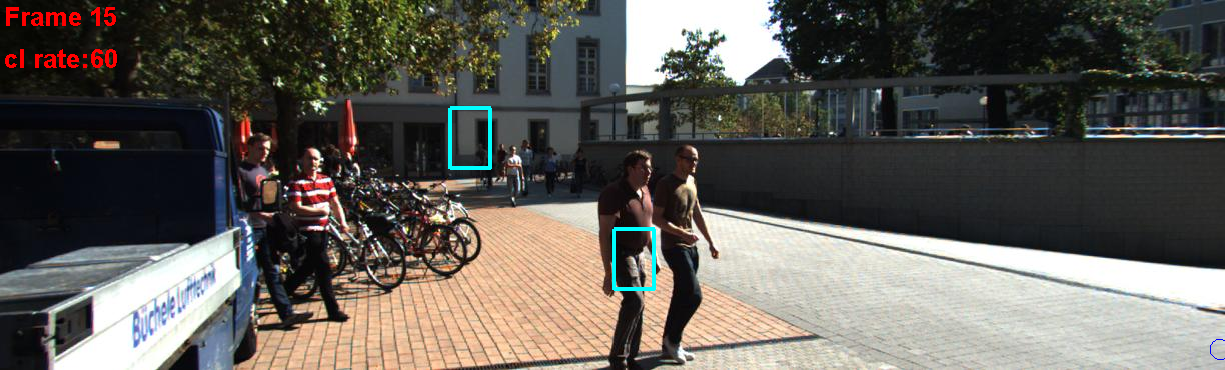}}%
\subfigure{
        \label{fig:subfig1}
    \includegraphics[scale=0.28, trim= 20 0 20 0]{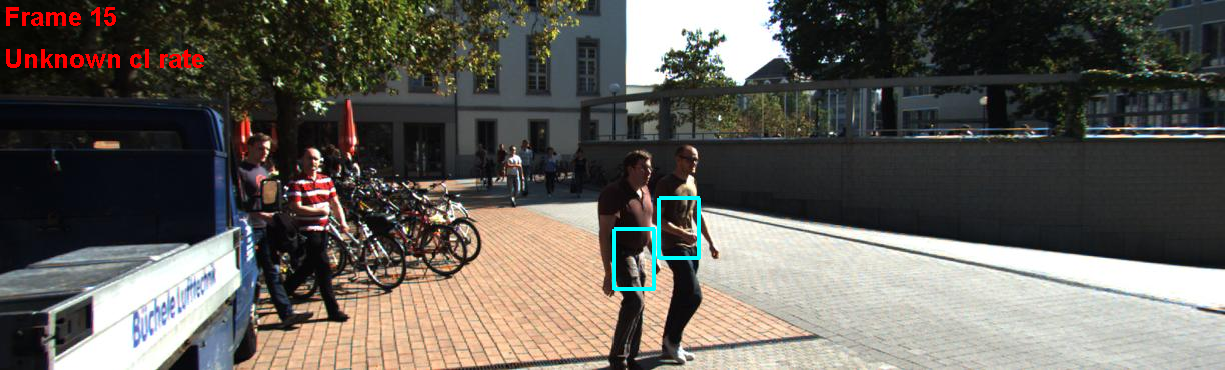}
} 
\subfigure{
        \label{fig:subfig1}
    \includegraphics[scale=0.28, trim= 20 0 20 0]{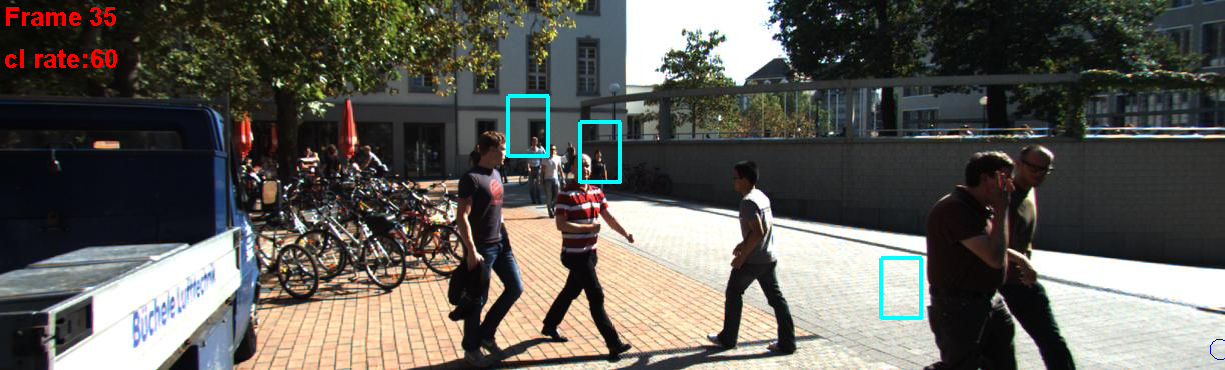}}%
\subfigure{
        \label{fig:subfig1}
    \includegraphics[scale=0.28, trim= 20 0 20 0]{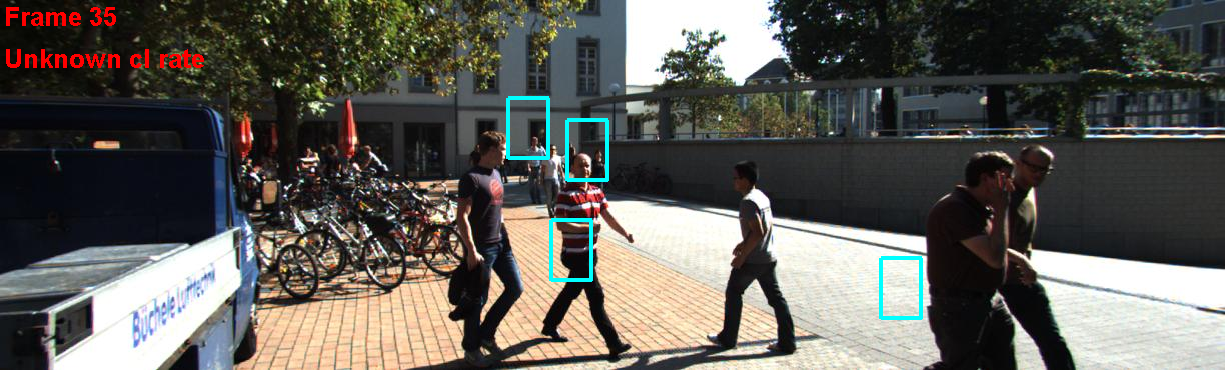}
} 
\subfigure{
        \label{fig:subfig5}
    \includegraphics[scale=0.28, trim= 20 0 20 0]{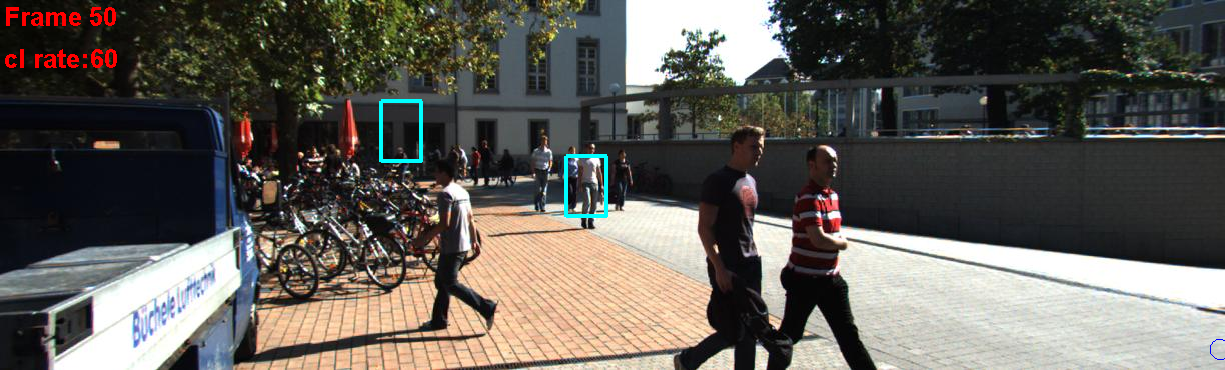}}%
\subfigure{
        \label{fig:subfig6}
    \includegraphics[scale=0.28, trim= 20 0 20 0]{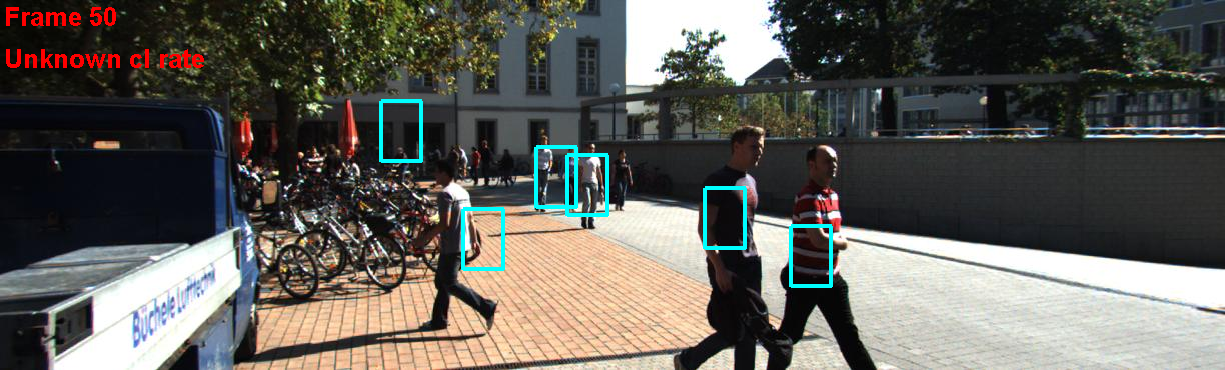}
} 
\caption[Examples of tracking results for frames 15,35,50 with guessed
clutter rate 60 (on left) and the proposed filter (on right)]{Tracking results for frames 15,35,50 with guessed clutter rate 60 (left) and the proposed filter (right) for dataset 2.}
\label{fig:videotrack2}
\end{figure*}

\begin{figure}[tbp]
\centering
\includegraphics[scale=0.6, trim= 35 0 20 20]{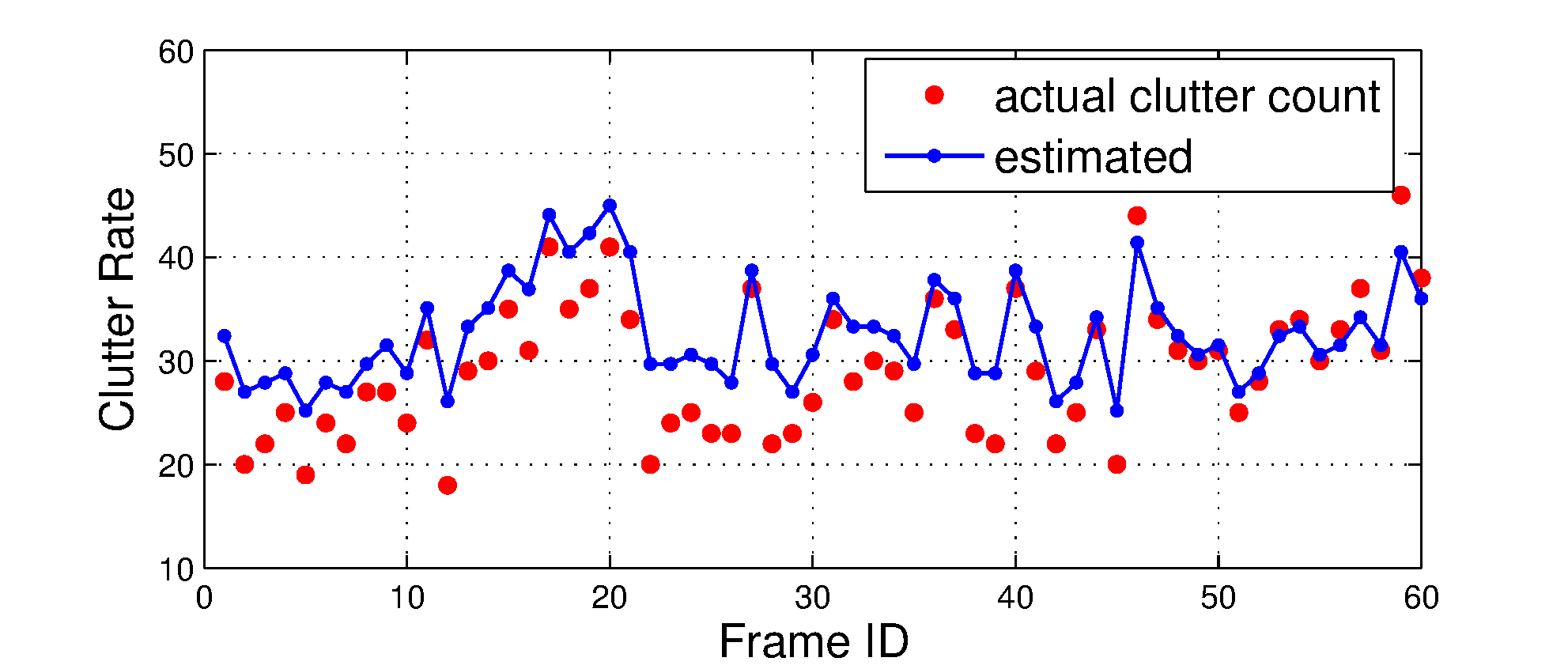}
\caption{Estimated clutter rate for dataset 2.}
\label{fig:video2clrate}
\end{figure}
\vspace{-2mm}
\section{Conclusion}

In this paper we have proposed a tractable algorithm for tracking multiple
objects in environments with unknown model parameters, such as clutter rate
and detection probability, based on the GLMB filter. Specifically, objects
of interest and clutter objects are treated as non-interacting classes of
objects, and a GLMB recursion for propagating the joint filtering density of
these classes are derived, along with an efficient implementation.
Simulations and applications to video data demonstrate that the proposed
filter has good tracking performance in the presence of unknown background
and outperforms the $\lambda $-CPHD filter. Moreover, it can also estimate
the clutter rate and detection probability parameters while tracking.

% trigger a \newpage just before the given reference
% number - used to balance the columns on the last page
% adjust value as needed - may need to be readjusted if
% the document is modified later
%\IEEEtriggeratref{41} % The "triggered" command can be changed if desired:
%%\IEEEtriggercmd{\enlargethispage{-5in}}
%
%% references section
%\bibliographystyle{IEEEtran}
%\bibliography{IEEEabrv,D:/Research/Tex/References/GLMB_ref}

% Generated by IEEEtran.bst, version: 1.13 (2008/09/30)
\providecommand{\url}[1]{#1} \csname url@samestyle\endcsname%
\providecommand{\newblock}{\relax} \providecommand{\bibinfo}[2]{#2} %
\providecommand{\BIBentrySTDinterwordspacing}{\spaceskip=0pt\relax} %
\providecommand{\BIBentryALTinterwordstretchfactor}{4} 
\providecommand{\BIBentryALTinterwordspacing}{\spaceskip=\fontdimen2\font plus
\BIBentryALTinterwordstretchfactor\fontdimen3\font minus
  \fontdimen4\font\relax} 
\providecommand{\BIBforeignlanguage}[2]{{\expandafter\ifx\csname l@#1\endcsname\relax
\typeout{** WARNING: IEEEtran.bst: No hyphenation pattern has been}\typeout{** loaded for the language `#1'. Using the pattern for}\typeout{** the default language instead.}\else
\language=\csname l@#1\endcsname
\fi
#2}} \providecommand{\BIBdecl}{\relax} \BIBdecl

\end{document}